\documentclass[
a4paper,reqno,10pt]{amsart}

\usepackage[top=2.5cm, bottom=2.5cm, left=2.5cm, right=2.5cm]{geometry}

\usepackage{amssymb,amsmath,amsthm,ascmac,array,bm}
\usepackage{color}
\usepackage{algorithm,algcompatible,amsmath}
\usepackage{booktabs}
\usepackage{float}
\usepackage{graphicx}
\usepackage{caption}
\usepackage{subcaption}

\usepackage[
colorlinks=true,bookmarks=true,
bookmarksnumbered=true,bookmarkstype=toc,linktocpage=true
]{hyperref}

\newtheorem{thm}{Theorem}[section]

\newtheorem{lem}[thm]{Lemma}

\newtheorem{rem}[thm]{Remark}

\numberwithin{equation}{section}
\allowdisplaybreaks

\newcommand{\mbbh}{\mathbb{H}}

\newcommand{\mbbr}{\mathbb{R}}

\newcommand{\al}{\alpha}  \newcommand{\ep}{\epsilon} 
 \newcommand{\del}{\delta}

\newcommand{\p}{\partial}
 
\newcommand{\argmax}{\mathop{\rm argmax}}
\newcommand{\norm}[1]{\left\lVert#1\right\rVert}
\newcommand{\cpoim}[1]{\tilde{N}(d#1,dy)} 
\newcommand{\poim}[1]{N(d#1,dy)}  
\algnewcommand\INPUT{\item[\textbf{Input:}]}%
\algnewcommand\OUTPUT{\item[\textbf{Output:}]}%
\newcommand\mymapsto{\mathrel{\ooalign{$\rightarrow$\cr%
			\kern-.15ex\raise.275ex\hbox{\scalebox{1}[0.522]{$\mid$}}\cr}}}

\def\nn{\nonumber}

 \def\E{\mathbb{E}}

\def\sumj{\sum_{j=1}^{n}}

\usepackage{mathtools}
\mathtoolsset{showonlyrefs=true}

\def\tes{\hat{\theta}_{n}}

\title[Quasi-likelihood-based EM algorithm for regime-switching SDE]
{
Quasi-likelihood-based EM algorithm for regime-switching SDE
}

\author{Yuzhong Cheng\textsuperscript{1}\textsuperscript{*}}\thanks{\textsuperscript{*}Corresponding author}
\address{\textsuperscript{1}(Corresponding author) Graduate School of Mathematics, Kyushu University, 744 Motooka Nishi-ku Fukuoka 819-0395, Japan}
\email{cheng.yuzhong.129@s.kyushu-u.ac.jp}
\author{Hiroki Masuda\textsuperscript{2}}
\address{\textsuperscript{2} Graduate School of Mathematical Sciences,\\ University of Tokyo, 3-8-1 Komaba Meguro-ku Tokyo 153-8914, Japan}
\email{hmasuda@ms.u-tokyo.ac.jp}

\date{\today}
\keywords{EM algorithm, Stochastic differential equations, Normal inverse Gaussian distribution, High-frequency sampling, Quasi-likelihood approach.
}

\begin{document}
\setlength{\baselineskip}{4.5mm}

\begin{abstract}
This paper considers estimating the parameters in a regime-switching stochastic differential equation(SDE) driven by Normal Inverse Gaussian(NIG) noise. The model under consideration incorporates a continuous-time finite state Markov chain to capture regime changes, enabling a more realistic representation of evolving market conditions or environmental factors. Although the continuous dynamics are typically observable, the hidden nature of the Markov chain introduces significant complexity, rendering standard likelihood-based methods less effective.  To address these challenges, we propose an estimation algorithm designed for discrete, high-frequency observations, even when the Markov chain is not directly observed. Our approach integrates the Expectation-Maximization (EM) algorithm, which iteratively refines parameter estimates in the presence of latent variables, with a quasi-likelihood method adapted to NIG noise. Notably, this method can simultaneously estimate parameters within both the SDE coefficients and the driving noise. Simulation results are provided to evaluate the performance of the algorithm. These experiments demonstrate that the proposed method provides reasonable estimation under challenging conditions.
\end{abstract}

\maketitle


\section{Introduction}
We consider the problem of parameter estimation for stochastic differential equations (SDEs) with Markovian switching, a class of hybrid systems that integrates continuous dynamics with discrete events. 
More precisely, It is a stochastic process with two components, one component $X$ describes the evolution of the process, and another component $\alpha$ describes the switching between different regimes. 
This framework is widely applicable in modeling complex systems where continuous processes interact with discrete state transitions, such as those encountered in finance, ecology, and engineering.

In this paper, we focus on applying an Expectation-Maximization (EM) algorithm to estimate parameters in Markov switching equations with hidden states. Besides the EM algorithm, another widely used approach to this problem is the Markov Chain Monte Carlo (MCMC) method, with relevant literature found in \cite{hahn2010markov} and \cite{hibbah2020mcmc}.

For the EM approach, a previous study by \cite{ChevallierGoutte2017} introduced a two-step method to estimate parameters in Markov switching stochastic differential equations with Lévy noise. Their approach first employs an EM algorithm designed to estimate parameters in a Markov switching SDE driven by Wiener noise, followed by a second step in which the parameters in the Lévy noise are estimated by fitting the noise distribution to each regime. Their study assumes fixed observation time steps.

In contrast, this paper also considers Markov switching SDEs with Lévy noise but proposes an EM algorithm that can simultaneously estimate parameters in both the drift coefficient and the Lévy noise components. Our method is based on a quasi-likelihood approach for Lévy noise under high-frequency observations, allowing for an adaptive, unfixed time step size.

Let $(\Omega,\mathcal{F},\{\mathcal{F}_t\},\mathbb{P})$ be a complete filtered probability space with the filtration $ \{\mathcal{F}_t\}_{t\ge 0}$ admitting usual conditions (i.e., it is right continuous and $\mathcal{F}_0$ contains all $\mathbb{P}$-null sets). The process $X = (X_t)_{t \geq 0}$ satisfies 
an $\mathbb{R}$-valued stochastic differential equation with switching given by the following expression
\begin{equation}
	dX_t = \lambda(b(\alpha_t)-X_t)dt +  dZ_t,
	\label{switchingsde}
\end{equation} while
$\alpha =(\alpha_t)_{t \geq 0}$ be a homogeneous continuous-time Markov chain with finite state space $\mathcal{S} = \{1,2,..., N\}$ where $N$ is a known positive integer. 
We always assume the Markov chain $\alpha$, the process $Z$ in \eqref{switchingsde} and $X_0$ are mutually independent.
The generator $Q = (q_{ij})$ of this Markov chain $\alpha$ is a $N \times N$ matrix whose elements are characterized by 
\begin{equation}
	\begin{array}{c}
		\mathbb{P}(\alpha_{t+h}=j \mid \alpha_t=i) 
		=\left\{\begin{array}{ll}
			q_{i j} h+o(h), & \text { if } i \neq j, \\
			1+q_{i i} h +o(h), & \text { if } i = j,
		\end{array}\right.
	\end{array}
	\label{CTMCgenerator}
\end{equation}
for $h \downarrow 0$. Here we have $q_{ij} > 0$ for $i\neq j$ and $q_{ii} = - \sum_{i \neq j}q_{ij}$ for all $i$.

The driving process $Z$ is a Normal Inverse Gaussian(NIG) L\'{e}vy process on $\mathbb{R}$ with $\mathcal{L}(Z_1) \sim \operatorname{NIG}(a,0,\delta,0)$. The density function of this NIG distribution is given by
	\begin{equation}
		f(z)= \frac{a \delta}{\pi} \exp \left\{\delta a\right\} \frac{K_{1}\left(a\sqrt{\delta^{2}+z^{2}}\right)}{\sqrt{\delta^{2}+z^{2}}},
		\label{nigdensity}
	\end{equation}
	where $\delta >0$, $a > 0$, $K_1$ is the modified Bessel function of third order with index 1. It is well known that NIG distribution is closed in convolution operation. Therefore we have $Z_t \sim \operatorname{NIG}(a,0,\delta t,0)$. For more properties of NIG distribution, see \cite{barndorff1997normal}. 

In the equation \eqref{switchingsde},  $\lambda \in \mathbb{R}_{+}$, and $b(\alpha_t)_{t \geq 0}$ is a continuous-time Markov chain
with state space $\{b(1),b(2),...,b(N)\}$, $b(i)\neq b(j)$ for $i \neq j$.
The existence and uniqueness of a strong solution of equation \eqref{switchingsde} is guaranteed by Theorem 2.1 in \cite{xi2017feller}.

In our case, $X$ is observed and $\alpha$ is latent (unobserved).
Let $h>0$, $t_j=jh$ for $j=0,1,...,n$ and $T_n = nh$. The observation is the vector $(X_{t_0},...,X_{t_n})$.
We should let $T_n$ be large enough to ensure that the solution of \eqref{switchingsde} switches through all states in the state space $\mathcal{S}$.
We assume the following high-frequency and long-term sampling conditions:
\begin{equation}\nn
\text{$T_n := nh \to \infty$\quad and \quad $h \to 0$\quad as $n \to \infty$.}    
\end{equation}


Let 
\begin{equation*}
    \theta = (b(1),b(2),...,b(N), \lambda, \delta),
\end{equation*} and $\theta \in \Theta$ where $\Theta = \prod_{i=1}^{N}\Theta_{b(i)}\times\Theta_{\lambda}\times\Theta_{\delta}$, and $N \geq 2$ be given.
All $\Theta_{\lambda}\subset \mathbb{R}_{+}, \Theta_{\delta}\subset \mathbb{R}_{+}, \Theta_{b(i)}\subset \mathbb{R}$ for $i =1,...,N$ are compact convex sets. Our objective is to estimate the parameter $\theta$ from $(X_{t_j})_{j=0}^{n}$. We should mention here that the parameter $a$ is a nuisance parameter, which may be unknown.
To simplify the computational complexity of the estimation procedure, we assume that 
$Q$ is known and fixed. The primary focus of this study is on the estimation of $\theta$. When the generator $Q$ is unknown, refer to Remark \ref{rem:unknownQ}.

The letter $C$ denotes a positive constant which may change from location to location.
We use $p_{\theta}$ as a generic symbol for densities and distributions parameterized by $\theta$. We will write $\mathbb{X}_{1,n} := (X_{t_1},...,X_{t_n})$, $\mathbb{A}_{1,n} := (\al_{t_1},...,\al_{t_n})$.

The paper is organized as follows: In Section \ref{sec:QLA}, we introduce the Cauchy quasi-likelihood approach to the NIG distribution. Section \ref{sec:EM} begins with a review of the standard EM algorithm, followed by our main contribution, proposing a numerical method based on an approximated EM algorithm, including methods for calculating the probability values required in the algorithm. In Section \ref{sec:sim}, we present numerical experiments to evaluate the performance of our method. Finally, an Appendix provides relevant background information.

\section{Quasi likelihood approach for NIG L\'{e}vy noise}
\label{sec:QLA}

In this section, we introduce a Cauchy quasi-likelihood function for the equation \eqref{switchingsde}, where the driving noise $Z$ is modeled as a NIG L\'{e}vy process.
Quasi-likelihood functions are well established in the literature and are an effective approach for approximating the true likelihood function. Gaussian quasi-likelihood functions are commonly used for models driven by a Wiener process, while for more general driving noises, such as jump processes, non-Gaussian quasi-likelihood functions may also be employed. This has been demonstrated to yield satisfactory statistical properties, as discussed in \cite{masuda2019non} and \cite{clement2020}.

When $Z$ is a NIG L\'{e}vy process where $Z_1$ follows $\operatorname{NIG}(a,0,\delta,0)$ where the density is given in \eqref{nigdensity}, we have that $Z_h$ and $\Delta_j Z$ follow the same distribution, and thus $Z_h$ follows an $\operatorname{NIG}(a,0,\delta h,0)$ distribution. Therefore, by Euler approximation, we have the approximate expression for $\Delta_jZ$:
\begin{equation*}
	X_{t_{j}} -\left(X_{t_{j-1}} +\lambda(b(\alpha_{t_{j-1}})-X_{t_{j-1}})h\right) \approx \Delta_jZ \sim \operatorname{NIG}(a,0,\delta h,0).
\end{equation*}
Also it is known that for any $c,d \in \mathbb{R}$ 
\begin{equation*}
	Y = cX+d \sim \operatorname{NIG}\left(\frac{a}{|c|},0,c\delta,d\right),
\end{equation*} when $X \sim \operatorname{NIG}(a,0,\delta,0)$, see  \cite{ChevallierGoutte2017}.
From this, we obtain that
\begin{equation*}
	\frac{\Delta_jZ}{\delta h} \sim \operatorname{NIG}(ah|\delta|,0,1,0).
\end{equation*} 
Substituting the approximate expression for $\Delta_jZ$ into the above relation, we have
\begin{equation*}
	\frac{X_{t_j}-\mu_{j-1}(\lambda)}{ \delta h} \approx \frac{\Delta_jZ}{\delta h} \sim \operatorname{NIG}(ah|\delta|,0,1,0),
\end{equation*}where 
\begin{equation*}
    \mu_{j-1}(\lambda) := X_{t_{j-1}} +\lambda(b(\alpha_{t_{j-1}})-X_{t_{j-1}})h.
\end{equation*}
As noted in \cite{kawai2013} this probability density converges to a standard Cauchy distribution as $h \to 0$. Thus, we can use the Cauchy distribution to approximate the NIG distribution for small-time step sizes. We let
$\boldsymbol{\eta}$ be a standard Cauchy random variable independent of $X$ and $\alpha$, then the conditional distribution of $X_{t_j}$ given $(X_{t_{j-1}},\alpha_{t_{j-1}})$ may be approximated by 
\begin{equation*}
	X_{t_j} \approx  \delta h \boldsymbol{\eta}  + \mu_{j-1}(\lambda) =  \delta h \boldsymbol{\eta}  + X_{t_{j-1}} +\lambda(b(\alpha_{t_{j-1}})-X_{t_{j-1}})h,
\end{equation*} when time stepsize $h$ be sufficiently small. We set $f(X_{t_j}|{X}_{t_{j-1}},\alpha_{t_{j-1}}; \theta)$ be a Cauchy density function given by 
\begin{equation}
	f(X_{t_j}|X_{t_{j-1}},\alpha_{t_{j-1}}; \theta)= \frac{1}{ \delta h\pi \left(1+ \left(\frac{X_{t_j}-\mu_{j-1}(\lambda)}{ \delta h}\right)^2\right)},
	\label{cauchytd}
\end{equation}
and now use this $f(X_{t_j}|X_{t_{j-1}},\alpha_{t_{j-1}}; \theta)$ to approximate transition density of $X_t$ from $t_{j-1}$ to $t_{j}$.
Observe that \eqref{nigdensity} contains the Bessel function, which makes it computationally demanding. A notable advantage of using the Cauchy distribution approximation is that it avoids this complexity.

Under the guidance of above ideology, we define Cauchy quasi-likelihood function $\mathbb{Q}_n(\theta)$ which is given by
\begin{equation}
	\mathbb{Q}_n(\theta) = \sumj \log f(X_{t_j}|X_{t_{j-1}},\alpha_{t_{j-1}}; \theta).
	\label{cauchyqlf}
\end{equation}
The Cauchy Quasi Maximum Likelihood Estimator (CQMLE) seeks to maximize the Cauchy quasi-likelihood function, $\mathbb{Q}_n(\theta)$, concerning the parameter $\theta$. However, it is not feasible to directly maximize $\mathbb{Q}_n(\theta)$ as it contains the unobserved process $\alpha$. Thus, function $\mathbb{Q}_n(\theta)$ here is not a proper choice in our situation. It should be noted that if $\alpha$ were observable, $\mathbb{Q}_n(\theta)$ would be a suitable choice for maximization. 
We propose an EM algorithm-based approach utilizing a Cauchy quasi-likelihood approximation in the next section.

\begin{rem}
    Since we use the Euler approximation to construct the quasi-likelihood function, our approach is not limited to the Ornstein-Uhlenbeck-type SDE \eqref{switchingsde}. In general, we can apply it to the broader class of regime-switching SDEs with the following form:
    \begin{equation*}
        dX_t = \mu(X_t,\alpha_t) dt + \sigma(X_t,\alpha_t) dZ_t.
    \end{equation*} See also Remark \ref{rem:cauchy}.
\end{rem}

\section{EM algorithm approach}
\label{sec:EM}

In this section, we consider the EM algorithm approach for equation \eqref{switchingsde}.

\subsection{Standard EM algorithm revisited}
We will use the EM algorithm in the next subsection. The EM algorithm introduced by \cite{dempster1977maximum} is an efficient iterative method to compute the Maximum Likelihood estimator when some variable is unobserved.
Here we introduce the basic background of the EM algorithm. Let $X$ and $Y$ be random variables, and define $p(x,y;\theta)$ as the joint density function of $(X, Y) \in \mathbb{R}^2$, and $p(y|x;\theta)$ as the conditional density function of $Y$ given $X$, where $\theta \in \Theta$, $\Theta$ is a bounded convex domain of parameters. Suppose that the variable $X$ is observed and $Y$ is a latent variable. The EM algorithm is to estimate $\theta$ under the incompletely observed data. The log-likelihood function $l_n(\theta)$ of $n$ i.i.d. observations $\{(X_1,..., X_n\}$ is defined by
\begin{equation*}
	l_n(\theta) = \sumj \log \left( \int p(X_j,y_j;\theta) dy_j \right).
\end{equation*}
The maximum likelihood estimator (MLE) is to find some $\hat{\theta}_n \in \Theta$ which maximizes the function $l_n(\theta)$. The EM algorithm uses an auxiliary function to calculate the MLE through iteration. We introduce a function $Q_n(\theta,\theta')$, which is called finite-sample Q-function, defined by
\begin{equation}
	Q_n(\theta,\theta') = \sumj  \left( \int p(y_j|X_j;\theta') \log p(X_j,y_j;\theta) dy_j \right).
	\label{FSQ}
\end{equation}

The EM algorithm consists of the following two steps. Given an initial value $\theta^{(0)}$
\begin{itemize}
	\item The first E-step is to compute the finite-sample Q-function $Q_n(\theta,\tes^{(m)})$,
	\begin{equation*}
		Q_n(\theta,\tes^{(m)}) = \sumj  \left( \int p(y_j|X_j;\tes^{(m)}) \log p(X_j,y_j;\theta) dy_j \right),
	\end{equation*} for $m = 0,1,...$;
	\item The second M-step is to find 
	\begin{equation*}
		\tes^{(m+1)} \in \argmax_{\theta} Q_n(\theta,\tes^{(m)}).
	\end{equation*}
\end{itemize}
After doing the E-step and M-step, we have the estimate $\tes^{(m+1)}$ through $\tes^{(m)}$. By repeating the E-step and M-step, the estimate $\tes^{(m+1)}$ is expected to converge to the MLE by the following reason.

Observe that, for any $\theta,\theta' \in \Theta$,
\begin{align*}
	l_n(\theta) &= \sumj \log \left( \int p(X_j,y_j;\theta) dy_j \right) \\
	&= \sumj \log \left( \int \frac{p(X_j,y_j;\theta)}{p(y_j|X_j;\theta')}p(y_j|X_j;\theta') dy_j\right)
	\\
	&\geq  \sumj   \int \log \left(\frac{p(X_j,y_j;\theta)}{p(y_j|X_j;\theta')}\right)p(y_j|X_j;\theta') dy_j
	\\
	&= \sumj  \int \log \left(p(X_j,y_j;\theta)\right)p(y_j|X_j;\theta') dy_j - \sumj  \int \log \left(p(y_j|X_j;\theta')\right)p(y_j|X_j;\theta') dy_j
	\\
	&=Q_n(\theta,\theta') - \sumj  \int \log \left(p(y_j|X_j;\theta')\right)p(y_j|X_j;\theta') dy_j,
\end{align*}
where the inequality follows from Jensen's inequality. It is clear to see that when $\theta = \theta'$ the inequality becomes an equality because $\frac{p(X_j,y_j;\theta')}{p(y_j|X_j;\theta')}= p(X_j;\theta')$. Let 
\begin{equation}
    \Delta_n(\theta ):= \sumj  \int \log \left(p(y_j|X_j;\theta)\right)p(y_j|X_j;\theta) dy_j,
    \nn
\end{equation}
so that we have 
\begin{equation}
	l_n(\theta) \geq Q_n(\theta,\theta')  - \Delta_n(\theta'),
	\label{EMineq}
\end{equation} for all $\theta$ and $\theta'$. Our goal here is to find a value $\theta$ that maximizes the log-likelihood function $\theta \mymapsto l_n(\theta)$. Now we have the function $Q_n(\theta,\theta')  - \Delta_n(\theta')$ is bounded from above by the log-likelihood function $l_n(\theta)$ and the relation $l_n(\theta') = Q_n(\theta',\theta')  - \Delta_n(\theta')$.
Hence we know that each $\theta$ maximizing the function $Q_n(\theta,\theta')  - \Delta_n(\theta')$ should increase the log-likelihood function $l_n(\theta)$. Thus, instead of directly maximizing function $l_n(\theta)$, we can try to maximize the function $\theta \mymapsto Q_n(\theta,\theta')  - \Delta_n(\theta')$. This is the basis of the Expectation-Maximization (EM) algorithm. We have 
\begin{align*}
	\tes^{(m+1)} &\in \argmax_{\theta} \left(Q_n(\theta,\tes^{(m)})  - \Delta_n(\tes^{(m)})\right)
	\\
	&\in \argmax_{\theta} Q_n(\theta,\tes^{(m)}),
\end{align*} where we drop the term which contains no $\theta$. Since
\begin{equation*}
Q_n(\tes^{(m+1)},\tes^{(m)})  - \Delta_n(\tes^{(m)}) \geq Q_n(\tes^{(m)},\tes^{(m)})  - \Delta_n(\tes^{(m)}),
\end{equation*} by \eqref{EMineq}, we have
\begin{equation*}
l_n(\tes^{(m+1)}) \geq l_n(\tes^{(m)}).
\end{equation*} This shows that the EM algorithm iteratively increases the log-likelihood function $l_n(\theta)$. 

The main advantage of the Expectation-Maximization (EM) algorithm is that it provides a convenient framework for solving Maximum Likelihood Estimation problems when unobserved or missing data is present. By maximizing the finite-sample Q-function instead of the log-likelihood function directly, the EM algorithm makes it easier to handle situations where some data is not directly available.

\subsection{EM algorithm for our Switching SDE}
Now, we shift our focus back to the switching SDE model. Our approach for estimating the parameter $\theta$ is built upon an EM algorithm tailored for diffusions with Markovian switching, as described in \cite{ChevallierGoutte2017}. However, our method diverges from \cite{ChevallierGoutte2017} in several key aspects. Specifically, we operate in a high-frequency observation setting, where the time step size $h$ is sufficiently small. Instead of relying on the commonly used Gaussian likelihood approximation, we leverage the asymptotic distribution of the Normal Inverse Gaussian (NIG) Lévy process over small time intervals to construct a Cauchy quasi-likelihood function.
As was mentioned, the EM algorithm is particularly well-suited for models involving latent variables, as it offers a computationally feasible approximation to the maximum likelihood estimator. A critical component of this method is the definition of an appropriate target function—in our case, a quasi-likelihood function—to be maximized through iterative steps of the EM algorithm.



Let $\mathbb{X}_{0,n} = (X_{t_0},X_{t_1},...,X_{t_n})$ denote all observations of $X$, $\mathbb{A}_{0,n}=(\al_{t_0},\al_{t_1},...,\al_{t_n})$. Define $p(\mathbb{X}_{0,n},\mathbb{A}_{0,n};\theta)$ as the joint density function of $\mathbb{X}_{0,n}$ and $\mathbb{A}_{0,n}$ under $\theta$ with respect to some $\sigma$-finite measure on $\mathbb{R}^{n+1}\times\mathcal{S}^{n+1}$. Since the continuous-time Markov chain $\alpha$ is not directly observed, we investigate an EM algorithm to estimate the parameters. Following the idea in the standard EM algorithm, we integrate  the joint density function $p(\mathbb{X}_{0,n},\mathbb{A}_{0,n};\theta)$ over the unobserved variables $\mathbb{A}_{0,n}$ to define the log-likelihood function $\mathbb{L}_n(\theta)$ as
\begin{align}
	\mathbb{L}_n(\theta) =& \log \left(\sum_{\mathbb{A}_{0,n} \in \mathcal{S}^n} p(\mathbb{X}_{0,n},\mathbb{A}_{0,n};\theta)\right).
\end{align} By following the standard EM algorithm framework, we derive the following: for any $\theta$ and $\theta'$,

\begin{align*}
	\mathbb{L}_n(\theta) =& \log \left(\sum_{\mathbb{A}_{0,n} \in \mathcal{S}^n} p(\mathbb{X}_{0,n},\mathbb{A}_{0,n};\theta)\right)
	\\
	=& \log \left(\sum_{\mathbb{A}_{0,n} \in \mathcal{S}^n} \frac{p(\mathbb{X}_{0,n},\mathbb{A}_{0,n};\theta)}{p\left(\mathbb{A}_{0,n}| \mathbb{X}_{0,n};\theta'\right)}p\left(\mathbb{A}_{0,n}| \mathbb{X}_{0,n};\theta'\right)\right)
	\\
	\geq& \sum_{\mathbb{A}_{0,n} \in \mathcal{S}^n} \log \left( \frac{p(\mathbb{X}_{0,n},\mathbb{A}_{0,n};\theta)}{p\left(\mathbb{A}_{0,n}| \mathbb{X}_{0,n};\theta'\right)}\right)p\left(\mathbb{A}_{0,n}| \mathbb{X}_{0,n};\theta'\right)
	\\
	=& \sum_{\mathbb{A}_{0,n} \in \mathcal{S}^n} \log \left( p(\mathbb{X}_{0,n},\mathbb{A}_{0,n};\theta)\right)p\left(\mathbb{A}_{0,n}| \mathbb{X}_{0,n};\theta'\right)
	\\
	&- \sum_{\mathbb{A}_{0,n} \in \mathcal{S}^n} \log \left( p\left(\mathbb{A}_{0,n}| \mathbb{X}_{0,n};\theta'\right)\right)p\left(\mathbb{A}_{0,n}| \mathbb{X}_{0,n};\theta'\right).
\end{align*}  

Within the standard EM algorithm framework, the function
\begin{equation*}
H_n(\theta,\theta'):= \sum_{\mathbb{A}_{0,n} \in \mathcal{S}^n} \log \left( p(\mathbb{X}_{0,n},\mathbb{A}_{0,n};\theta)\right)p\left(\mathbb{A}_{0,n}| \mathbb{X}_{0,n};\theta'\right)
\end{equation*}
 is commonly used for iterative optimization. However, in our context, the explicit form of the function $p(\mathbb{X}_{0,n},\mathbb{A}_{0,n};\theta)$ is unknown, making a direct use in the EM algorithm unsuitable. To address this issue, we introduce a quasi-likelihood approximation for $p(\mathbb{X}_{0,n},\mathbb{A}_{0,n};\theta)$ as a feasible alternative. By employing the quasi-likelihood approximation, the function  $H_n(\theta,\theta')$ can be reformulated to accommodate its use within the EM algorithm framework.

We now detail the procedure to modify the function $H_n(\theta,\theta')$. 
First, observe that
\begin{align*}
	p(\mathbb{X}_{0,n},\mathbb{A}_{0,n};\theta) = p(\mathbb{X}_{0,n} | \mathbb{A}_{0,n};\theta) p(\mathbb{A}_{0,n};\theta).
\end{align*} 
Recall that the Markov chain $\alpha$ and the L\'evy process $Z$ are assumed to be independent. Consequently, the process $X$ can be viewed as a solution to an inhomogeneous Markovian L\'evy-driven stochastic differential equation, conditional on $\alpha$. The Euler approximation for $X_{t_j}$ is expressed as
\begin{align*}
	X_{t_{j}} \approx X_{t_{j-1}} +\lambda(b(\alpha_{t_{j-1}})-X_{t_{j-1}})h +  \Delta_j Z.
\end{align*}
This approximation provides an approximated representation for the density of $X_{t_j}$ conditional on $X_{t_{j-1}}$ and $\alpha_{t_{j-1}}$, which can be modeled by the Cauchy density function $f(X_{t_j}|X_{t_{j-1}},\alpha_{t_{j-1}};\theta)$ as in \eqref{cauchytd}. Hence, we approximate $p(\mathbb{X}_{0,n} | \mathbb{A}_{0,n};\theta)$ by the following:
\begin{align}
	p(\mathbb{X}_{0,n} | \mathbb{A}_{0,n};\theta) &= p(X_{t_n}|\mathbb{X}_{0,n-1},\mathbb{A}_{0,n};\theta) p(\mathbb{X}_{0,n-1}|\mathbb{A}_{0,n};\theta)
    \notag\\
    &\approx p(X_{t_n}|X_{t_{n-1}},\alpha_{t_{n-1}};\theta)p(\mathbb{X}_{0,n-1}|\mathbb{A}_{0,n-1};\theta)
	\notag\\
    &\approx f(X_{t_n}|X_{t_{n-1}},\alpha_{t_{n-1}};\theta)p(\mathbb{X}_{0,n-1}|\mathbb{A}_{0,n-1};\theta).\notag
\end{align}
We employ the following two approximation steps: 
\begin{itemize}
	\item In the first approximation, we approximate $p(\mathbb{X}_{0,n} | \mathbb{A}_{0,n};\theta)$ as the product of $p(X_{t_j}|X_{t_{j-1}},\alpha_{t_{j-1}};\theta)$, utilizing the Markov property of $X_t$ conditioned on $\alpha_t$.
	\item In the second approximation, we approximate $p(X_{t_j}|X_{t_{j-1}},\alpha_{t_{j-1}};\theta)$ using the Cauchy density $f(X_{t_j}|X_{t_{j-1}},\alpha_{t_{j-1}};\theta)$ in \eqref{cauchytd}. This approximation is particularly useful in high-frequency scenarios and serves as the foundation for our quasi-likelihood approach to approximating $p(\mathbb{X}_{0,n}|\mathbb{A}_{0,n};\theta)$.
\end{itemize}
This two-step approximation allows us to construct the following quasi-likelihood approach suitable for high-frequency observations.
By induction, we derive:
\begin{equation}
    p(\mathbb{X}_{0,n} | \mathbb{A}_{0,n};\theta) \approx \prod_{j=1}^{n}f(X_{t_j}|X_{t_{j-1}},\alpha_{t_{j-1}};\theta)
	\label{quasitran}.
\end{equation}
Additionally, note that $p(\mathbb{A}_{0,n};\theta)$ can be expressed as the product of transition densities $p(\alpha_{t_j}|\alpha_{t_{j-1}};\theta)$. 
Consequently, we approximate $p(\mathbb{A}_{0,n};\theta)$ as:
\begin{align}
	p(\mathbb{A}_{0,n};\theta) & = \prod_{j=1}^{n} p(\alpha_{t_j}|\alpha_{t_{j-1}};\theta) 
	\notag\\
	& = \prod_{j=1}^{n} \left( I_{\{\alpha_{t_j}=\alpha_{t_{j-1}}\}}\left(1+q_{\alpha_{t_{j-1}}\alpha_{t_{j-1}}}h +o(h)\right) + I_{\{\alpha_{t_j}\neq \alpha_{t_{j-1}}\}}\left(q_{\alpha_{t_{j-1}}\alpha_{t_j}} h+o(h)\right) \right)
	\notag\\
	& \approx  \prod_{j=1}^{n} \left( I_{\{\alpha_{t_j}=\alpha_{t_{j-1}}\}}\left(1+q_{\alpha_{t_{j-1}}\alpha_{t_{j-1}}}h \right) + I_{\{\alpha_{t_j}\neq \alpha_{t_{j-1}}\}}q_{\alpha_{t_{j-1}}\alpha_{t_j}}h  \right),
	\label{approximatealpha}
\end{align}
where $I$ is the indicator function. 
The Cauchy quasi-likelihood approximation of $p(\mathbb{X}_{0,n},\mathbb{A}_{0,n};\theta)$ combines \eqref{quasitran} and \eqref{approximatealpha}, leading to:
\begin{align}
	p&(\mathbb{X}_{0,n},\mathbb{A}_{0,n};\theta) =  p(\mathbb{X}_{0,n} | \mathbb{A}_{0,n};\theta) p(\mathbb{A}_{0,n};\theta)
	\notag\\
	&\approx \prod_{j=1}^{n}
	f(X_{t_j}|X_{t_{j-1}},\alpha_{t_{j-1}};\theta) \left(I_{\{\alpha_{t_j}=\alpha_{t_{j-1}}\}}\left(1+q_{\alpha_{t_{j-1}}\alpha_{t_{j-1}}}h \right) + I_{\{\alpha_{t_j}\neq \alpha_{t_{j-1}}\}}q_{\alpha_{t_{j-1}}\alpha_{t_j}}h\right).
\end{align}
Using this approximation, it is natural to define the modified function $\mbbh_n(\theta,\theta')$ as follows:
\begin{align}
	\mbbh_n(\theta;\theta') = \sumj \sum_{i \in \mathcal{S}} \sum_{k \in \mathcal{S}} \log \biggl\{&f\left(X_{t_j} | X_{t_{j-1}}, \alpha_{t_{j-1}}=i; \theta\right) 
	 \left(I_{\{k=i\}}\left(1+q_{ii}h\right) + I_{\{k\neq i\}}q_{ik}h\right)
	\biggr\}
	\notag\\
	& \cdot \mathbb{P}\left(\alpha_{t_{j-1}}=i,\alpha_{t_j}=k| \mathbb{X}_{0,n};\theta'\right),
	\label{emcauchyqlf}
\end{align} with which we will replace $H_n(\theta,\theta')$ in the EM algorithm.

In $\mbbh_n(\theta;\theta')$, the conditional probabilities of all latent variables $\mathbb{A}_{0,n}$ in $H_n(\theta,\theta')$ are reduced to the conditional probabilities of two consecutive variables $\alpha_{t_{j-1}}$ and $\alpha_{t_j}$. This modification aligns the logarithmic terms in the summation and simplifies the computational process.
The components of $\mbbh_n(\theta;\theta')$, except for $\mathbb{P}\left(\alpha_{t_{j-1}}=i,\alpha_{t_{j}}=k| \mathbb{X}_{0,n};\theta'\right)$, are given explicitly.
This facilitates differentiation with respect to the parameter $\theta$, which is essential for optimization in the EM algorithm.
The term $\mathbb{P}\left(\alpha_{t_{j-1}}=i,\alpha_{t_{j}}=k| \mathbb{X}_{0,n};\theta'\right)$ is computable using numerical methods, which will be discussed in Section \ref{subsec:prob_s}.


\begin{rem}
\label{rem:cauchy}
We should mention it here that in the second step of approximation \eqref{quasitran}, we utilize the Cauchy approximation $f\left(X_{t_j} | X_{t_{j-1}}, \alpha_{t_{j-1}}=i; \theta\right)$ due to the Normal Inverse Gaussian (NIG) distributional property of the driving Lévy noise. In high-frequency settings, the NIG density can be effectively approximated by a Cauchy density.
In general, an appropriate quasi-likelihood should be selected according to the specific distributional characteristics of the driving L\'{e}vy process. As long as the high-frequency sampling $h\to 0$ for infinite-activity non-Gaussian $Z$ is concerned, one natural (and often optimal) choice is the $\al$-stable approximation for $\al\in(0,2)$. We refer interested readers to \cite{masuda2019non}, \cite{Mas23}, \cite{CleGlo20}, and the references therein.
\end{rem}

To formulate the EM algorithm based on the modified function $\mbbh_n(\theta;\theta')$, we proceed as follows. 
\begin{description}
    \item[\textbf{Initialization}]
    Begin by specifying an initial value for the parameter vector:
    \begin{equation*}
        \hat{\theta}^{(0)} = (b(1)^{(0)},b(2)^{(0)},...,b(N)^{(0)},\lambda^{(0)}, \delta^{(0)},Q^{(0)}),
    \end{equation*} where $\hat{\theta}^{(0)}$ serves as the starting point for the iterative procedure.
    \item[\textbf{First E-step}]
    In the first E-step, calculate the function $\mbbh_n(\theta;\theta')$ by substituting the initial parameter values $\hat{\theta}^{(0)}$ into $\theta'$:
    \begin{align*}
 	\mbbh_n(\theta;\hat{\theta}^{(0)}) = \sumj \sum_{i \in \mathcal{S}} \sum_{k \in \mathcal{S}} \log \biggl\{&f\left(X_{t_j} | X_{t_{j-1}}, \alpha_{t_{j-1}}=i; \theta\right) 
 	\left(I_{\{k=i\}}\left(1+q_{ii}h \right) + I_{\{k\neq i\}}q_{ik}h\right)
 	\biggr\}
 	\notag\\
 	& \cdot \mathbb{P}\left(\alpha_{t_{j-1}}=i,\alpha_{t_j}=k| \mathbb{X}_{0,n};\hat{\theta}^{(0)}\right),
    \end{align*}
    \item[\textbf{First M-step}] In the first M-step, maximize $\theta \mymapsto \mbbh_n(\theta;\hat{\theta}^{(0)})$ to obtain an updated parameter estimate:
    \begin{equation*}
	\hat{\theta}^{(1)}  \in \argmax_{\theta} \mbbh_n(\theta;\hat{\theta}^{(0)}).
    \end{equation*}
    \item[\textbf{Iterative Steps}] After obtaining $\hat{\theta}^{(1)}$, set $\theta' = \hat{\theta}^{(1)}$ in equation \eqref{emcauchyqlf}, and repeat the above mentioned E- and M-steps iteratively.
    Suppose we are at the $m+1$-th step:
\end{description}

\begin{description}
	\item[\textbf{E-step}] the finite-sample Cauchy quasi-likelihood function for $m+1$ step is 
	\begin{align}
	\mbbh_n(\theta;\hat{\theta}^{(m)}) 
 = \sumj \sum_{i \in \mathcal{S}} \sum_{k \in \mathcal{S}} \log \biggl\{&f\left(X_{t_j} | X_{t_{j-1}}, \alpha_{t_{j-1}}=i; \theta\right) 
\left(I_{\{k=i\}}\left(1+q_{ii}h \right) + I_{\{k\neq i\}}q_{ik}h\right)
\biggr\}
\notag\\
& \cdot \mathbb{P}\left(\alpha_{t_{j-1}}=i,\alpha_{t_j}=k| \mathbb{X}_{0,n};\hat{\theta}^{(m)}\right).
		\label{estep}
	\end{align}
   \item[\textbf{M-step}] M-step is to compute the maximizer of $\mbbh_n(\theta;\hat{\theta}^{(m)})$. We let
   \begin{equation*}
   	\hat{\theta}^{(m+1)} \in \argmax_{\theta} \mbbh_n(\theta;\hat{\theta}^{(m)}).
   \end{equation*}
\end{description} 

Note that we have 
\begin{align*}
	\\
	&\mbbh_n(\hat{\theta}^{(m+1)};\hat{\theta}^{(m)})  \geq \mbbh_n(\hat{\theta}^{(m)};\hat{\theta}^{(m)}) .
\end{align*}

To compute $\hat{\theta}$ in each step, for example at $m+1$ step, we compute the derivatives of $\mbbh_n(\theta;\hat{\theta}^{(m)})$ with respect to $\theta$. First, we have
\begin{align*}
		\mbbh_n(\theta;\hat{\theta}^{(m)}) = \sumj \sum_{i \in \mathcal{S}}  \sum_{k \in \mathcal{S}} \log &\left( \frac{1}{ \delta h\pi \left(1+ \left(\frac{X_{t_j}-\mu_{j-1}(\lambda)}{ \delta h}\right)^2\right)}\left(I_{\{k=i\}}\left(1+q_{ii}h \right) + I_{\{k\neq i\}}q_{ik}h\right)\right) 
		\\
		&\cdot \mathbb{P}\left(\alpha_{t_{j-1}}=i,\alpha_{t_j}=k| \mathbb{X}_{0,n};\hat{\theta}^{(m)}\right) 
		\\
		= \sumj \sum_{i \in \mathcal{S}} \log &\left( \frac{1+q_{ii}h}{ \delta h\pi \left(1+ \left(\frac{X_{t_j}-\mu_{j-1}(\lambda)}{ \delta h}\right)^2\right)}\right) \mathbb{P}\left(\alpha_{t_{j-1}}=i,\alpha_{t_j}=i| \mathbb{X}_{0,n};\hat{\theta}^{(m)}\right)
		\\
		+ \sumj \sum_{\substack{k\neq i\\k,i \in \mathcal{S}}}& \log \left( \frac{q_{ik} h}{ \delta h\pi \left(1+ \left(\frac{X_{t_j}-\mu_{j-1}(\lambda)}{ \delta h}\right)^2\right)}\right) \mathbb{P}\left(\alpha_{t_{j-1}}=i,\alpha_{t_j}=k| \mathbb{X}_{0,n};\hat{\theta}^{(m)}\right).
\end{align*}
 Here we give expressions of these derivatives. First, we have
\begin{align*}
	\p_{\delta} \mbbh_n(\theta;\hat{\theta}^{(m)}) =  \sumj \sum_{i \in \mathcal{S}} \sum_{k \in \mathcal{S}}&\left(K_{1,j}(\theta)K_{2,j}(\theta)\right) 
\\
&\cdot \mathbb{P}\left(\alpha_{t_{j-1}}=i,\alpha_{t_j}=k| \mathbb{X}_{0,n};\hat{\theta}^{(m)}\right) 
\end{align*}
where
\begin{align*}
	&K_{1,j}(\theta) := \frac{\pi(X_{t_j}-\mu_{j-1}(\lambda))^2}{\delta^2 h } -  h \pi, \\
	&K_{2,j}(\theta) := \frac{1}{ \delta h\pi \left(1+ \left(\frac{X_{t_j}-\mu_{j-1}(\lambda)}{ \delta h}\right)^2\right)}.
\end{align*}
We have
\begin{align*}
\p_{\lambda} \mbbh_n(\theta;\hat{\theta}^{(m)}) = \sumj \sum_{i \in \mathcal{S}} \sum_{k \in \mathcal{S}}&\left(K_{2,j}(\theta)K_{3,j}(\theta) \right)
\\
& \cdot \mathbb{P}\left(\alpha_{t_{j-1}}=i,\alpha_{t_j}=k| \mathbb{X}_{0,n};\hat{\theta}^{(m)}\right) 
\end{align*}
where
\begin{align*}
	K_{3,j}(\theta) := \frac{2\pi(X_{t_j}-\mu_{j-1}(\lambda))(b(i)-X_{t_{j-1}})}{ \delta}.
\end{align*}
For the derivatives with respect to $b(l)$, $l \in \mathcal{S}$, we have
\begin{align*}
\p_{b(l)} \mbbh_n(\theta;\hat{\theta}^{(m)}) = \sumj \sum_{k \in \mathcal{S}} (K_{2,l,j}^{\star}(\theta)K_{4,l,j}(\theta)) \mathbb{P}\left(\alpha_{t_{j-1}}=l,\alpha_{t_j}=k| \mathbb{X}_{0,n};\hat{\theta}^{(m)}\right)
\end{align*}
where
\begin{align*}
	&K_{2,l,j}^{\star}(\theta) := \frac{1}{ \delta h\pi \left(1+ \left(\frac{X_{t_j}-(X_{t_{j-1}} +\lambda(b(l)-X_{t_{j-1}})h)}{ \delta h}\right)^2\right)},
	\\
	&K_{4,l,j}(\theta) := \frac{2\pi(X_{t_j}-(X_{t_{j-1}} +\lambda(b(l)-X_{t_{j-1}})h))\lambda}{ \delta}.
\end{align*}

In many applications of the EM algorithm, it is impossible to explicitly find out the M-step. As can be seen from the above expressions of the first derivative, $\hat{\theta}^{(m+1)}$ can not be computed in a closed form. Thus we have to apply some feasible numerical approaches.
One such approach is solving the M-step by using a one-step estimation method of the Newton-Raphson type, as described in \cite{lange1995gradient}.  In this case the estimator $\hat{\theta}^{(m+1)}$ is given by
\begin{align}
	\hat{\theta}^{(m+1)}_{\star} := \hat{\theta}^{(m)} - \left(\p^2_{\theta}\mbbh_n(\hat{\theta}^{(m)};\hat{\theta}^{(m)})\right)^{-1} \p_{\theta} \mbbh_n(\hat{\theta}^{(m)};\hat{\theta}^{(m)}).
	\label{newton}
\end{align}
This method suggests that a one-step Newton method at each M-step would be sufficient to ensure the convergence of the EM algorithm.
To avoid potential computational issues, we will not make direct use of the Hessian matrix $\p^2_{\theta}\mbbh_n(\hat{\theta}^{(m)};\hat{\theta}^{(m)})$; just for reference, we give its expression in Section \ref{sec:2nd.order.derivatives}.

Another approach is the first-order EM algorithm introduced in \cite{balakrishnan2017statistical}.
Let $\hat{\theta}^{(m+1)}_{\star,\rho}$ be the estimate in $m+1$ step computed by the first-order method:
\begin{align}
\hat{\theta}^{(m+1)}_{\star,\rho} = \hat{\theta}^{(m)} + \rho \p_{\theta}\mbbh_n(\theta = \hat{\theta}^{(m)};\hat{\theta}^{(m)}),
\label{onestep}
\end{align} where $\rho > 0$ is an appropriately chosen value. In the paper \cite{balakrishnan2017statistical}, the value of $\rho$ is chosen such that it satisfies certain conditions that guarantee the theoretical convergence of the EM algorithm. 
Through these fomula we can nurmerically compute $\theta = \hat{\theta}^{(m+1)}$ without solving the equation $\p_{\theta} \mbbh_n(\theta;\hat{\theta}^{(m)})=0$. If we use the first-order method in \eqref{onestep}, when we enter the $(m+2)$th iteration, we let $\hat{\theta}^{(m+1)} = \hat{\theta}^{(m+1)}_{\star,\rho}$, so the new M-step in $(m+1)$th step becomes
\begin{description}
	\item [\textbf{M$^\prime$-step}] Compute $\hat{\theta}^{(m+1)}_{\star,\rho}$ using \eqref{onestep}
	\begin{equation*}
		\hat{\theta}^{(m+1)}_{\star,\rho} = \hat{\theta}^{(m)} + \rho \p_{\theta}\mbbh_n(\theta = \hat{\theta}^{(m)};\hat{\theta}^{(m)}),
	\end{equation*} and then let $\hat{\theta}^{(m+1)} = \hat{\theta}^{(m+1)}_{\star,\rho}$ for the $m+2$ step.
\end{description}

We propose the EM algorithm as in Algorithm \ref{AL:EM}.
\begin{algorithm}
	\caption{EM Algorithm for switching SDE}
	\begin{algorithmic}[1]
		\STATE Choose initial value $\hat{\theta}^{(0)} = (b(1)^{(0)},b(2)^{(0)},...,b(N)^{(0)},\lambda^{(0)}, \delta^{(0)})$, a small positive  $\ep$, and an appropriate positive $\rho$. Let $m=0$.
		\STATE \textbf{E-step:} Compute the function $\mbbh_n(\theta;\hat{\theta}^{(m)})$  where
		\begin{align*}
			\mbbh_n(\theta;\hat{\theta}^{(m)}) 
			= \sumj \sum_{i \in \mathcal{S}} \sum_{k \in \mathcal{S}} \log \biggl\{&f\left(X_{t_j} | X_{t_{j-1}}, \alpha_{t_{j-1}}=i; \theta\right) 
			\left(I_{\{k=i\}}\left(1+q_{ii} h\right) + I_{\{k\neq i\}}q_{ik}h\right)
			\biggr\}
			\notag\\
			& \cdot \mathbb{P}\left(\alpha_{t_{j-1}}=i,\alpha_{t_j}=k| \mathbb{X}_{0,n};\hat{\theta}^{(m)}\right) 
		\end{align*}
		\STATE \textbf{M$^\prime$-step}: Compute the estimate $\hat{\theta}^{(m+1)}_{\star,\rho} = \hat{\theta}^{(m)} $ where
		\begin{equation*}
			\hat{\theta}^{(m+1)}_{\star,\rho} = \hat{\theta}^{(m)} + \rho \p_{\theta}\mbbh_n(\theta= \hat{\theta}^{(m)};\hat{\theta}^{(m)}).
		\end{equation*}
		\STATE Let $\hat{\theta}^{(m+1)} = \hat{\theta}^{(m+1)}_{\star,\rho}$.
		\STATE Check the termination condition. If the difference in absolute value is less than $\ep$, then stop the algorithm and output $\hat{\theta}^{(m+1)}$ as the result. If not, continue to the next step.
		\STATE Let $m = m + 1$ and return to the E-step.
	\end{algorithmic}
\label{AL:EM}
\end{algorithm}
The computation of $\mathbb{P}\left(\alpha_{t_{j-1}}=i,\alpha_{t_j}=k| \mathbb{X}_{0,n};\hat{\theta}^{(m)}\right)$ is presented in Algorithm \ref{AL:prob}.

\begin{rem}
\label{rem:unknownQ}
    In this study, to focus on our primary objective, namely the estimation of $\theta$, we are assuming that the generator $Q$ is known. 
    If computational complexity is not a concern, the EM algorithm could be extended to treat $Q$ as an additional estimation target.
    To incorporate $Q$ into the estimation, the Cauchy quasi-likelihood function \eqref{emcauchyqlf} in the EM algorithm can be formally modified as follows:
\begin{align*}
    \mbbh_n(\theta,Q;\hat{\theta}^{(m)},\hat{Q}^{(m)}) 
 = \sumj \sum_{i \in \mathcal{S}} \sum_{k \in \mathcal{S}} \log \biggl\{&f\left(X_{t_j} | X_{t_{j-1}}, \alpha_{t_{j-1}}=i; \theta\right) 
\left(I_{\{k=i\}}\left(1+q_{ii}h \right) + I_{\{k\neq i\}}q_{ik}h\right)
\biggr\}
\notag\\
& \cdot \mathbb{P}\left(\alpha_{t_{j-1}}=i,\alpha_{t_j}=k| \mathbb{X}_{0,n};\hat{\theta}^{(m)},\hat{Q}^{(m)}\right). \nn
\end{align*}
In addition, the M-step must include an update for $Q$, which can be performed as follows in an analogous way to \eqref{onestep}:
\begin{equation*}
    \hat{Q}^{(m+1)} = \hat{Q}^{(m)} + \rho \p_Q \mbbh_n(\hat{\theta}^{(m)},\hat{Q}^{(m)};\hat{\theta}^{(m)},\hat{Q}^{(m)}).
\end{equation*}
The first-order derivatives of $\mbbh_n(\theta,Q;\hat{\theta}^{(m)},\hat{Q}^{(m)})$ with respect to $q_{ik}$ are given by:
\begin{align*}
	\p_{q_{ll}} \mbbh_n(\theta,Q;\hat{\theta}^{(m)},\hat{Q}^{(m)}) &= \sumj  \frac{h}{1+ q_{ll}h}\mathbb{P}\left(\alpha_{t_{j-1}}=l,\alpha_{t_j}=l| \mathbb{X}_{0,n};\hat{\theta}^{(m)},\hat{Q}^{(m)}\right),
	\\
	\p_{q_{lm}}\mbbh_n(\theta,Q;\hat{\theta}^{(m)},\hat{Q}^{(m)}) &= \sumj  \frac{1}{q_{lm}}\mathbb{P}\left(\alpha_{t_{j-1}}=l,\alpha_{t_j}=m| \mathbb{X}_{0,n};\hat{\theta}^{(m)},\hat{Q}^{(m)}\right), \,\,\, \operatorname{for} l \neq m.
\end{align*}
The computation of the probabilities $\mathbb{P}\left(\alpha_{t_{j-1}}=i,\alpha_{t_j}=k| \mathbb{X}_{0,n};\hat{\theta}^{(m)},\hat{Q}^{(m)}\right)$ is analogous to the method described in Section \ref{subsec:prob_s}, with replace $q_{ik}$ by $\hat{q}_{ik}^{(m)}$.
\end{rem}

\subsection{Caculation of some probabilities}

\label{subsec:prob_s}

The iteration procedure in \eqref{estep} involves probability values $\mathbb{P}\left(\alpha_{t_{j-1}}=i,\alpha_{t_j}=k| \mathbb{X}_{0,n};\hat{\theta}^{(m)}\right) $ for  $i,k \in \mathcal{S}$ and $j = 1,2,...,n$. However, computing these probabilities can be challenging as these values depend on all information $\mathbb{X}_{0,n}$. To address this, we can adopt a method proposed by \cite{kim1994dynamic}, which was originally used for dynamic linear models, and adapt it for our purposes.
In this subsection, we illustrate the method for computing these values.
The calculation of $\mathbb{P}\left(\alpha_{t_{j-1}}=i,\alpha_{t_j}=k| \mathbb{X}_{0,n};\hat{\theta}^{(m)}\right)$ can be divided into three steps:

\begin{description}
    \item[\textbf{Step 1}] Compute the filtered probability  
    $\mathbb{P}\left(\alpha_{t_{j}}=i \mid \mathbb{X}_{0,j};\hat{\theta}^{(m)}\right).$ \item[\textbf{Step 2}] Using the results obtained in Step 1, calculate the  probability  
    $\mathbb{P}\left(\alpha_{t_{j-1}}=i, \alpha_{t_{j}} = k \mid \mathbb{X}_{0,j-1};\hat{\theta}^{(m)}\right)$  
    and the probability  
    $\mathbb{P}\left(\alpha_{t_{j}}=k \mid \mathbb{X}_{0,j-1};\hat{\theta}^{(m)}\right).$ \item[\textbf{Step 3}] Utilize the results from Step 2 to compute the joint probability  
    $\mathbb{P}\left(\alpha_{t_{j}}=i, \alpha_{t_{j+1}}=k \mid \mathbb{X}_{0,n};\hat{\theta}^{(m)}\right).$  
\end{description}

We will describe each part in detail and also note that, in some cases, the explicit computation of these probabilities may be difficult due to numerical issues, so we may need to use approximations instead of explicit values. By following this procedure, we can obtain reliable probability values for the iterative process. Let $p(\cdot)$ be the density function of $\mathbb{P}(\cdot)$ with respect to suitable measure.

\subsubsection{Step 1 : Calculation of \, $ \mathbb{P}\left(\alpha_{t_{j}}=i |\mathbb{X}_{0,j};\hat{\theta}^{(m)}\right)$ }

We have 
\begin{align*}
	\mathbb{P}\left(\alpha_{t_{j}}=i | \mathbb{X}_{0,j};\hat{\theta}^{(m)}\right) 
	&=  \frac{\sum_{k \in \mathcal{S}}p\left(X_{t_j},\alpha_{t_{j}}=i , \alpha_{t_{j-1}}=k  |  \mathbb{X}_{0,j-1}; \hat{\theta}^{(m)}\right)}{p\left(X_{t_j} |  \mathbb{X}_{0,j-1}; \hat{\theta}^{(m)}\right)}.
\end{align*}
Note that we have
\begin{align*}
	p\left(X_{t_j} |  \mathbb{X}_{0,j-1}; \hat{\theta}^{(m)}\right) =  \sum_{l \in \mathcal{S}}\sum_{k \in \mathcal{S}}p\left(X_{t_j},\alpha_{t_{j}}=l, \alpha_{t_{j-1}} = k,  |  \mathbb{X}_{0,j-1}; \hat{\theta}^{(m)}\right),
\end{align*} so it is enough to compute $p\left(X_{t_j},\alpha_{t_{j}}=i, \alpha_{t_{j-1}}=k  |  \mathbb{X}_{0,j-1}; \hat{\theta}^{(m)}\right)$. We propose to use  the following expression: 
\begin{align}
	&p\left(X_{t_j},\alpha_{t_{j}}=i, \alpha_{t_{j-1}}=k  |  \mathbb{X}_{0,j-1}; \hat{\theta}^{(m)}\right)
	\notag\\
	&= p\left(X_{t_j} |\alpha_{t_{j}}=i, \alpha_{t_{j-1}} = k,  \mathbb{X}_{0,j-1}; \hat{\theta}^{(m)}\right) \mathbb{P}\left(\alpha_{t_{j}}=i, \alpha_{t_{j-1}}=k | \mathbb{X}_{0,j-1}; \hat{\theta}^{(m)}\right)
	\notag\\
	&\approx  f(X_{t_j}|X_{t_{j-1}},\alpha_{t_{j-1}}=k; \hat{\theta}^{(m)}) \mathbb{P}\left(\alpha_{t_{j}}=i, \alpha_{t_{j-1}}=k | \mathbb{X}_{0,j-1}; \hat{\theta}^{(m)}\right)
	\notag\\
	& = f(X_{t_j}|X_{t_{j-1}},\alpha_{t_{j-1}}=k; \hat{\theta}^{(m)}) \mathbb{P}\left(\alpha_{t_{j}} = i |  \alpha_{t_{j-1}}=k,\mathbb{X}_{0,j-1}  ;\hat{\theta}^{(m)}\right) \mathbb{P}\left(\alpha_{t_{j-1}}=k|  \mathbb{X}_{0,j-1};\hat{\theta}^{(m)}\right)
	\notag\\
	& = f(X_{t_j}| X_{t_{j-1}},\alpha_{t_{j-1}}=k; \hat{\theta}^{(m)}) 
	\left( I_{\{k=i\}}\left(1+q_{kk} h\right) + I_{\{k\neq i\}}q_{ki}h\right)
	\mathbb{P}\left(\alpha_{t_{j-1}}=k|  \mathbb{X}_{0,j-1};\hat{\theta}^{(m)}\right).
	\label{approxtrans}
\end{align}
Notice that we used the approximation in the third line of \eqref{approxtrans} as follows:
\begin{equation}
    p\left(X_{t_j} |\alpha_{t_{j}}=i,  \alpha_{t_{j-1}} = k,  \mathbb{X}_{0,j-1}; \hat{\theta}^{(m)}\right) \approx f(X_{t_j}|X_{t_{j-1}}, \alpha_{t_{j-1}}=k; \hat{\theta}^{(m)}),
    \label{eq:pdensity}
\end{equation}
and also used the following approximation in the line of \eqref{approxtrans}:
\begin{align}
	\mathbb{P}\left(\alpha_{t_{j}} = k |  \alpha_{t_{j-1}}=i, \mathbb{X}_{0,j-1}  ;\hat{\theta}^{(m)}\right) 
	\approx I_{\{k=i\}}\left(1+q_{ii} h\right) + I_{\{k\neq i\}}q_{ik}h.
	\label{eq:tran,j-1}
\end{align}
We explain the reason why \eqref{eq:pdensity} holds as follows:


\begin{align}
	p\left(X_{t_j}|\alpha_{t_{j}}=i, \alpha_{t_{j-1}}=k,  \mathbb{X}_{0,j-1}; \hat{\theta}^{(m)}\right) &= \frac{p\left(X_{t_j},\alpha_{t_{j}}=i| \alpha_{t_{j-1}}=k,  \mathbb{X}_{0,j-1}; \hat{\theta}^{(m)}\right) }{\mathbb{P}\left(\alpha_{t_{j}}=i| \alpha_{t_{j-1}}=k,  \mathbb{X}_{0,j-1}; \hat{\theta}^{(m)}\right)}
	\notag\\
	&=\frac{p\left(X_{t_j},\alpha_{t_{j}}=i| \alpha_{t_{j-1}}=k,  X_{t_{j-1}}; \hat{\theta}^{(m)}\right) }{\mathbb{P}\left(\alpha_{t_{j}}=i| \alpha_{t_{j-1}}=k,  \mathbb{X}_{0,j-1}; \hat{\theta}^{(m)}\right)}
	\notag\\
	& \approx \frac{f(X_{t_j}|X_{t_{j-1}},\alpha_{t_{j-1}}=k,  \hat{\theta}^{(m)})\mathbb{P}\left(\alpha_{t_{j}}=i| \alpha_{t_{j-1}}=k,  X_{t_{j-1}}; \hat{\theta}^{(m)}\right) }{\mathbb{P}\left(\alpha_{t_{j}}=i| \alpha_{t_{j-1}}=k,  \mathbb{X}_{0,j-1}; \hat{\theta}^{(m)}\right)}
	\notag\\
	&= f(X_{t_j}|X_{t_{j-1}},\alpha_{t_{j-1}}=k; \hat{\theta}^{(m)}).
	\label{approreason2}
\end{align}
 Note that we used the Markov property of the two-component process $(X_t,\alpha_t)$ in the second line of \eqref{approreason2} and that we applied our Cauchy approximation toward the unknown transition probability of $(X_t,\alpha_t)$ in the third line of \eqref{approreason2}.

So we finally obtain a formula to compute $\mathbb{P}\left(\alpha_{t_{j}}=i | \mathbb{X}_{0,j}\hat{\theta}^{(m)}\right)$, that is 
\begin{align*}
	\mathbb{P}\left(\alpha_{t_{j}}=i | \mathbb{X}_{0,j};\hat{\theta}^{(m)}\right)  
	=
	 \frac{\sum_{k \in \mathcal{S}}p\left(X_{t_j},\alpha_{t_{j}}=i , \alpha_{t_{j-1}}=k  |  \mathbb{X}_{0,j-1}; \hat{\theta}^{(m)}\right)}{\sum_{l \in \mathcal{S}}\sum_{k \in \mathcal{S}}p\left(X_{t_j},\alpha_{t_{j}}=l, \alpha_{t_{j-1}} = k,  |  \mathbb{X}_{0,j-1}; \hat{\theta}^{(m)}\right)},
\end{align*}
with the approximation
\begin{align*}
	&p\left(X_{t_j},\alpha_{t_{j}}=i, \alpha_{t_{j-1}}=k  |  \mathbb{X}_{0,j-1}; \hat{\theta}^{(m)}\right) 
	\\
	&\approx  f(X_{t_j}|X_{t_{j-1}},\alpha_{t_{j-1}}=k;  \hat{\theta}^{(m)}) 
	\left( I_{\{k=i\}}\left(1+q_{kk} h\right) + I_{\{k\neq i\}}q_{ki}h\right)
	\mathbb{P}\left(\alpha_{t_{j-1}}=k|  \mathbb{X}_{0,j-1};\hat{\theta}^{(m)}\right).
\end{align*}
Note that the value $\mathbb{P}\left(\alpha_{t_{j}}=i | \mathbb{X}_{0,j};\hat{\theta}^{(m)}\right)$ for each $j$ is  obtained from the value $\mathbb{P}\left(\alpha_{t_{j-1}}=k|  \mathbb{X}_{0,j-1};\hat{\theta}^{(m)}\right)$. Therefore we obtain a forward algorithm to calculate $\mathbb{P}\left(\alpha_{t_{j}}=i |\mathbb{X}_{0,j};\hat{\theta}^{(m)}\right)$ for all $j = 1,...,n$ with given initial values $\mathbb{P}\left(\alpha_{t_{0}}=i |X_{0};\hat{\theta}^{(m)}\right)$ for $i \in \mathcal{S}$.


\subsubsection{Step 2 : Calculation of \, $\mathbb{P}\left(\alpha_{t_{j-1}}=i, \alpha_{t_{j}} = k |  \mathbb{X}_{0,j-1};\hat{\theta}^{(m)}\right)$ \, and \, $\mathbb{P}\left(\alpha_{t_{j}}=k |\mathbb{X}_{0,j-1};\hat{\theta}^{(m)}\right)$}

By straightforward computations and approximation \eqref{eq:tran,j-1}, we have 
\begin{align}
	\mathbb{P}\left(\alpha_{t_{j-1}}=i, \alpha_{t_{j}} = k |  \mathbb{X}_{0,j-1};\hat{\theta}^{(m)}\right) 
	&= 
	\mathbb{P}\left(\alpha_{t_{j}} = k |  \alpha_{t_{j-1}}=i,\mathbb{X}_{0,j-1}  ;\hat{\theta}^{(m)}\right) \mathbb{P}\left(\alpha_{t_{j-1}}=i|  \mathbb{X}_{0,j-1};\hat{\theta}^{(m)}\right)
    \nn \\
    & \approx \left(I_{\{k=i\}}\left(1+q_{ii} h\right) + I_{\{k\neq i\}}q_{ik}h\right) \mathbb{P}\left(\alpha_{t_{j-1}}=i|  \mathbb{X}_{0,j-1};\hat{\theta}^{(m)}\right), \nn
\end{align}
and 
\begin{align*}
	\mathbb{P}\left(\alpha_{t_{j}}=k |\mathbb{X}_{0,j-1};\hat{\theta}^{(m)}\right) &= \sum_{i \in \mathcal{S}}\mathbb{P}\left(\alpha_{t_{j}}=k,\alpha_{t_{j-1}}=i |\mathbb{X}_{0,j-1};\hat{\theta}^{(m)}\right)
	\\
	&\approx \sum_{i \in \mathcal{S}} \left(I_{\{k=i\}}\left(1+q_{ii} h\right) + I_{\{k\neq i\}}q_{ik}h\right) \mathbb{P}\left(\alpha_{t_{j-1}}=i |\mathbb{X}_{0,j-1};\hat{\theta}^{(m)}\right). 
\end{align*}
Note that $ \mathbb{P}\left(\alpha_{t_{j-1}}=i |\mathbb{X}_{0,j-1};\hat{\theta}^{(m)}\right)$ is obtained in Step 1.


\subsubsection{Step 3 : Calculation of \, $ \mathbb{P}\left(\alpha_{t_{j}}=i, \alpha_{t_{j+1}}=k | \mathbb{X}_{0,n};\hat{\theta}^{(m)}\right)$}

We first propose the following approximation:
\begin{align}
	 &\mathbb{P}\left(\alpha_{t_{j}}=i, \alpha_{t_{j+1}}=k | \mathbb{X}_{0,n};\hat{\theta}^{(m)}\right) \nn\\
	&=\mathbb{P}\left(\alpha_{t_{j+1}}=k | \mathbb{X}_{0,n};\hat{\theta}^{(m)}\right) \mathbb{P}\left(\alpha_{t_{j}}=i | \alpha_{t_{j+1}}=k,  \mathbb{X}_{0,n};\hat{\theta}^{(m)}\right) 
	\notag\\
	&\approx  \mathbb{P}\left(\alpha_{t_{j+1}}=k | \mathbb{X}_{0,n};\hat{\theta}^{(m)}\right) \mathbb{P}\left(\alpha_{t_{j}}=i | \alpha_{t_{j+1}}=k,  \mathbb{X}_{0,j};\hat{\theta}^{(m)}\right) 
	\notag\\
	&= \frac{\mathbb{P}\left(\alpha_{t_{j+1}}=k | \mathbb{X}_{0,n};\hat{\theta}^{(m)}\right) \mathbb{P}\left(\alpha_{t_{j}}=i,  \alpha_{t_{j+1}}=k | \mathbb{X}_{0,j};\hat{\theta}^{(m)}\right) }{\mathbb{P}\left(\alpha_{t_{j+1}}=k | \mathbb{X}_{0,j};\hat{\theta}^{(m)}\right)}.
	\label{appronj}
\end{align}
Here, we have used the following approximation in the second line:
\begin{equation}\label{hm:add-1}
    \mathbb{P}\left(\alpha_{t_{j}}=i | \alpha_{t_{j+1}}=k,  \mathbb{X}_{0,n};\hat{\theta}^{(m)}\right) \approx \mathbb{P}\left(\alpha_{t_{j}}=i | \alpha_{t_{j+1}}=k,  \mathbb{X}_{0,j};\hat{\theta}^{(m)}\right).
\end{equation}
Below, we illustrate the reason for the approximation \eqref{hm:add-1}.
Note that 
\begin{align}
	& \mathbb{P}\left(\alpha_{t_{j}}=i | \alpha_{t_{j+1}}=k,  \mathbb{X}_{0,n};\hat{\theta}^{(m)}\right) \nn\\
    &= \mathbb{P}\left(\alpha_{t_{j}}=i | \alpha_{t_{j+1}}=k,  \mathbb{X}_{0,j}, \mathbb{X}_{j+1,n};\hat{\theta}^{(m)}\right)
	\nn\\
	&= \frac{ p\left(\alpha_{t_{j}}=i , \mathbb{X}_{j+1,n}| \alpha_{t_{j+1}}=k,  \mathbb{X}_{0,j};\hat{\theta}^{(m)}\right)}{  p\left(\mathbb{X}_{j+1,n}| \alpha_{t_{j+1}}=k, \mathbb{X}_{0,j}, ;\hat{\theta}^{(m)}\right)}
	\nn\\
	&= \frac{ \mathbb{P}\left(\alpha_{t_{j}}=i | \alpha_{t_{j+1}}=k,  \mathbb{X}_{0,j};\hat{\theta}^{(m)}\right)  p\left( \mathbb{X}_{j+1,n}| \alpha_{t_{j}}=i , \alpha_{t_{j+1}}=k,  \mathbb{X}_{0,j};\hat{\theta}^{(m)}\right)}{  p\left(\mathbb{X}_{j+1,n}| \alpha_{t_{j+1}}=k, \mathbb{X}_{0,j}, ;\hat{\theta}^{(m)}\right)}.
    \label{apprea}
\end{align}
Now if the equation
\begin{align}
	p\left( \mathbb{X}_{j+1,n}| \alpha_{t_{j}}=i , \alpha_{t_{j+1}}=k,  \mathbb{X}_{0,j};\hat{\theta}^{(m)}\right) = p\left(\mathbb{X}_{j+1,n}| \alpha_{t_{j+1}}=k, \mathbb{X}_{0,j}, ;\hat{\theta}^{(m)}\right)
	\label{approiftrue}
\end{align}
holds, then the rightmost side of \eqref{apprea} simplifies to $\mathbb{P}\left(\alpha_{t_{j}}=i | \alpha_{t_{j+1}}=k,  \mathbb{X}_{0,j};\hat{\theta}^{(m)}\right)$, leading to \eqref{hm:add-1} and hence 
the approximation \eqref{appronj}. In this case,
this equation implies that knowing both $\mathbb{X}_{0,j}$ and $\alpha_{t_{j+1}}$, $\alpha_{t_j}$ will not provide any additional information on $ \mathbb{X}_{j+1,n}$ beyond what is already contained in $\mathbb{X}_{0,j}$ and $\alpha_{t_{j+1}}$. However, it should be noted that the argument is an approximate one since \eqref{approiftrue} may not always be true.

The expression \eqref{appronj} give us a backward method to compute $ \mathbb{P}\left(\alpha_{t_{j}}=i, \alpha_{t_{j+1}}=k | \mathbb{X}_{0,n};\hat{\theta}^{(m)}\right) $, knowing $\mathbb{P}\left(\alpha_{t_{n}}=i | \mathbb{X}_{0,n};\hat{\theta}^{(m)}\right)$, $ \mathbb{P}\left(\alpha_{t_{j}}=i,  \alpha_{t_{j+1}}=k | \mathbb{X}_{0,j};\hat{\theta}^{(m)}\right)$ and $\mathbb{P}\left(\alpha_{t_{j+1}}=k | \mathbb{X}_{0,j};\hat{\theta}^{(m)}\right)$ for $i,k \in \mathcal{S}$, $j=0,1,2,...,n-1$. 
We illustrate here the computation steps:
\begin{enumerate}
	\item for $j=n-1$, $i,k \in \mathcal{S}$:
	\begin{align*}
		\mathbb{P}\left(\alpha_{t_{n-1}}=i, \alpha_{t_{n}}=k | \mathbb{X}_{0,n};\hat{\theta}^{(m)}\right) \approx \frac{\mathbb{P}\left(\alpha_{t_{n}}=k | \mathbb{X}_{0,n};\hat{\theta}^{(m)}\right) \mathbb{P}\left(\alpha_{t_{n-1}}=i,  \alpha_{t_{n}}=k | \mathbb{X}_{0,n-1};\hat{\theta}^{(m)}\right) }{\mathbb{P}\left(\alpha_{t_{n}}=k | \mathbb{X}_{0,n-1};\hat{\theta}^{(m)}\right)}.
	\end{align*}
	\item for $j=n-2$, $i,k \in \mathcal{S}$:
	\begin{align*}
		\mathbb{P}\left(\alpha_{t_{n-2}}=i, \alpha_{t_{n-1}}=k | \mathbb{X}_{0,n};\hat{\theta}^{(m)}\right) \approx \frac{\mathbb{P}\left(\alpha_{t_{n-1}}=k | \mathbb{X}_{0,n};\hat{\theta}^{(m)}\right) \mathbb{P}\left(\alpha_{t_{n-2}}=i,  \alpha_{t_{n-1}}=k | \mathbb{X}_{0,n-2};\hat{\theta}^{(m)}\right) }{\mathbb{P}\left(\alpha_{t_{n-1}}=k | \mathbb{X}_{0,n-2};\hat{\theta}^{(m)}\right)},
	\end{align*}
	where 
	\begin{align*}
		\mathbb{P}\left(\alpha_{t_{n-1}}=k | \mathbb{X}_{0,n};\hat{\theta}^{(m)}\right)  = \sum_{l \in \mathcal{S}} \mathbb{P}\left(\alpha_{t_{n-1}}=k, \alpha_{t_{n}}=l | \mathbb{X}_{0,n};\hat{\theta}^{(m)}\right). 
	\end{align*}
	Note that the values $\mathbb{P}\left(\alpha_{t_{n-1}}=k, \alpha_{t_{n}}=l | \mathbb{X}_{0,n};\hat{\theta}^{(m)}\right)$ for all $k,l \in \mathcal{S}$ are provided by the previous step (1).
	\item Repeat this process several times until we obtain the values $\mathbb{P}\left(\alpha_{t_{j}}=i, \alpha_{t_{j+1}}=k | \mathbb{X}_{0,n};\hat{\theta}^{(m)}\right) $ for all $j=0,1,...,n-1$ and $i,k \in \mathcal{S}$.
\end{enumerate}


\medskip

Finally, we summarize the computation method in Algorithm \ref{AL:prob}.
By combining Algorithms \ref{AL:EM} and \ref{AL:prob}, we can now use our method to estimate the parameters.

\begin{algorithm}
	\caption{Computation of $\mathbb{P}\left(\alpha_{t_{j-1}}=i,\alpha_{t_j}=k| \mathbb{X}_{0,n};\hat{\theta}^{(m)}\right) $ for $i,k \in \mathcal{S}$}
	\begin{algorithmic}[1]
		\STATE Choose initial values $\mathbb{P}\left(\alpha_{t_{0}}=i |\mathbb{X}_{0,0};\hat{\theta}^{(m)}\right)$ for all $i \in \mathcal{S}$. Let $j=1$.
		\WHILE {$j \leq n$}
		\STATE Compute $\mathbb{P}\left(\alpha_{t_{j}}=i | \mathbb{X}_{0,j};\hat{\theta}^{(m)}\right) $ for $i \in \mathcal{S}$ as
	\begin{align*}
		\mathbb{P}\left(\alpha_{t_{j}}=i | \mathbb{X}_{0,j};\hat{\theta}^{(m)}\right)  
		=
		\frac{\sum_{k \in \mathcal{S}}p\left(X_{t_j},\alpha_{t_{j}}=i , \alpha_{t_{j-1}}=k  |  \mathbb{X}_{0,j-1}; \hat{\theta}^{(m)}\right)}{\sum_{l \in \mathcal{S}}\sum_{k \in \mathcal{S}}p\left(X_{t_j},\alpha_{t_{j}}=l, \alpha_{t_{j-1}} = k,  |  \mathbb{X}_{0,j-1}; \hat{\theta}^{(m)}\right)},
	\end{align*}
	with
	\begin{align*}
		p&\left(X_{t_j},\alpha_{t_{j}}=i, \alpha_{t_{j-1}}=k  |  \mathbb{X}_{0,j-1}; \hat{\theta}^{(m)}\right) 
		\\
		&\approx  f(X_{t_j}|\alpha_{t_{j-1}}=k, X_{t_{j-1}}, \hat{\theta}^{(m)}) 
		\left( I_{\{k=i\}}\left(1+q_{kk} h\right) + I_{\{k\neq i\}}q_{ki}h\right)
		\mathbb{P}\left(\alpha_{t_{j-1}}=k|  \mathbb{X}_{0,j-1};\hat{\theta}^{(m)}\right).
	\end{align*}
	\STATE $j = j+1$.
	\ENDWHILE
	\STATE Let $j=1$.
			\WHILE {$j \leq n$}
		\STATE Compute $\mathbb{P}\left(\alpha_{t_{j-1}}=i, \alpha_{t_{j}} = k |  \mathbb{X}_{0,j-1};\hat{\theta}^{(m)}\right)$ and $\mathbb{P}\left(\alpha_{t_{j}}=k |\mathbb{X}_{0,j-1};\hat{\theta}^{(m)}\right)$ for $i,k \in \mathcal{S}$ as
	\begin{align*}
		\mathbb{P}\left(\alpha_{t_{j-1}}=i, \alpha_{t_{j}} = k |  \mathbb{X}_{0,j-1};\hat{\theta}^{(m)}\right) 
		= 
		\mathbb{P}\left(\alpha_{t_{j}} = k |  \alpha_{t_{j-1}}=i,\mathbb{X}_{0,j-1}  ;\hat{\theta}^{(m)}\right) \mathbb{P}\left(\alpha_{t_{j-1}}=i|  \mathbb{X}_{0,j-1};\hat{\theta}^{(m)}\right)
	\end{align*}
	and 
	\begin{align*}
		\mathbb{P}\left(\alpha_{t_{j}}=k |\mathbb{X}_{0,j-1};\hat{\theta}^{(m)}\right) &= \sum_{l \in \mathcal{S}}\mathbb{P}\left(\alpha_{t_{j}}=k,\alpha_{t_{j-1}}=l |\mathbb{X}_{0,j-1};\hat{\theta}^{(m)}\right)
		\\
		&= \sum_{l \in \mathcal{S}} \mathbb{P}\left(\alpha_{t_{j}}=k |\alpha_{t_{j-1}}=l, \mathbb{X}_{0,j-1};\hat{\theta}^{(m)}\right) \mathbb{P}\left(\alpha_{t_{j-1}}=l |\mathbb{X}_{0,j-1};\hat{\theta}^{(m)}\right), 
	\end{align*}
    respectively, with
	\begin{align*}
		\mathbb{P}\left(\alpha_{t_{j}} = k |  \alpha_{t_{j-1}}=i, \mathbb{X}_{0,j-1}  ;\hat{\theta}^{(m)}\right) = \mathbb{P}\left(\alpha_{t_{j}} = k |  \alpha_{t_{j-1}}=i  ;\hat{\theta}^{(m)}\right) 
		\approx I_{\{k=i\}}\left(1+q_{ii} h\right) + I_{\{k\neq i\}}q_{ik}h.
	\end{align*}
	\STATE $j=j+1$
	\ENDWHILE
	\STATE Let $j=n$.
	\WHILE {$j \geq 1$}
	\STATE Compute $\mathbb{P}\left(\alpha_{t_{j-1}}=i,\alpha_{t_j}=k| \mathbb{X}_{0,n};\hat{\theta}^{(m)}\right) $ for $i,k \in \mathcal{S}$ as
		\begin{align*}
		\mathbb{P}\left(\alpha_{t_{j-1}}=i,\alpha_{t_j}=k| \mathbb{X}_{0,n};\hat{\theta}^{(m)}\right)  
		\approx
		 \frac{\mathbb{P}\left(\alpha_{t_{j}}=k | \mathbb{X}_{0,n};\hat{\theta}^{(m)}\right) \mathbb{P}\left(\alpha_{t_{j-1}}=i,  \alpha_{t_{j}}=k | \mathbb{X}_{0,j-1};\hat{\theta}^{(m)}\right) }{\mathbb{P}\left(\alpha_{t_{j}}=k | \mathbb{X}_{0,j-1};\hat{\theta}^{(m)}\right)}.
		\end{align*}
		\STATE $j = j-1$.
		\ENDWHILE
	\end{algorithmic}
	\label{AL:prob}
\end{algorithm}




\section{Simulation study}
\label{sec:sim}

In this section, we present the simulation results obtained through the application of our proposed EM algorithm. The Euler approximation method, with a small step size, was utilized to generate sample paths of our switching stochastic differential equation model. The algorithms described in Section \ref{sec:EM} were then applied to estimate the parameters of the model. In the subsequent subsections, we provide a detailed description of the parameter settings, sample generation procedure, choice of initial values, step size, and terminal time in our simulation. The estimated values of the parameters are presented in table form, and visual representations in the form of figures are provided to illustrate the intermediate values of the estimator during the iterative procedure.

In the simulation, we set the state space $\mathcal{S}=\{1,2\}$. This means the equation \eqref{switchingsde} is switching between two states. We consider the following setting of the parameters:
\begin{align*}
	&Q = 
	\begin{pmatrix}
		-0.009& 0.009  \\
		0.005 & -0.005\\
	\end{pmatrix},
\\
&b(1)  = 6, \quad b(2) = 3,
\\
& \lambda = 2, \quad \delta=1, \quad a =0.3.
\end{align*}

Our estimation target is parameter $\theta = (b(1),b(2),\lambda,\delta)$.


\begin{figure}
	\centering
	\begin{subfigure}{.4\textwidth}
		\includegraphics[width=\textwidth]{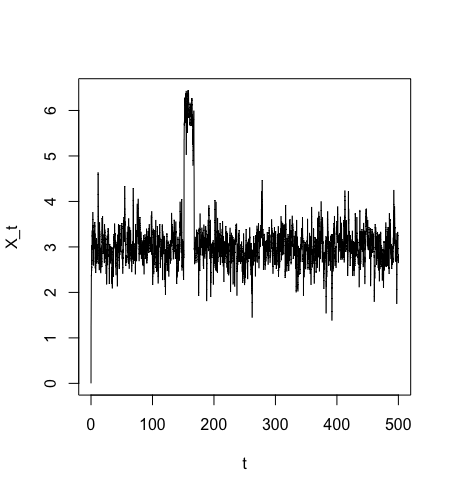}
		\caption{Sample path with $T = 500$, $h=0.1$}
		\label{500_5000}
	\end{subfigure}
	\begin{subfigure}{.4\textwidth}
		\includegraphics[width=\textwidth]{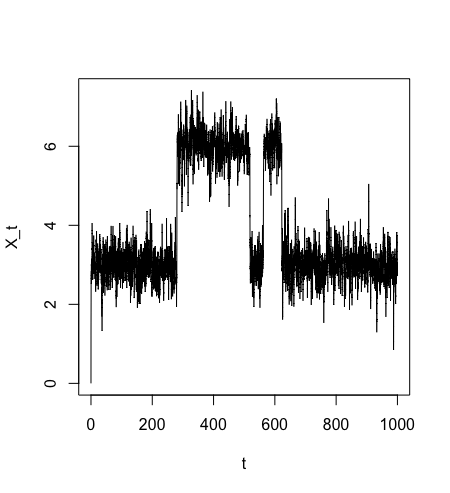}
		\caption{Sample path with $T = 1000$, $h=0.1$}
		\label{1000_10000}
	\end{subfigure}
	\begin{subfigure}{.4\textwidth}
		\includegraphics[width=\textwidth]{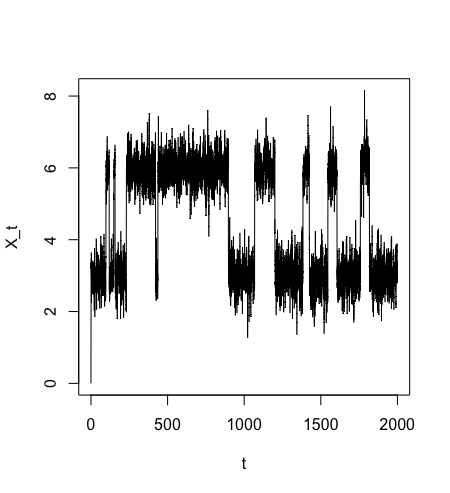}
		\caption{Sample path with $T = 2000$, $h=0.1$}
		\label{2000_20000}
	\end{subfigure}
	\begin{subfigure}{.4\textwidth}
		\includegraphics[width=\textwidth]{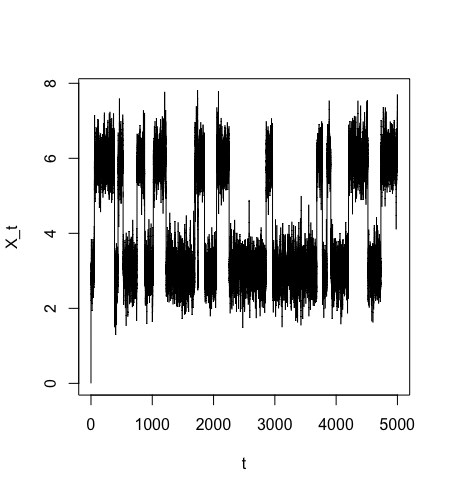}
		\caption{Sample path with $T = 5000$, $h=0.1$}
		\label{5000_50000}
	\end{subfigure}
	\caption{Sample paths}
	\label{sample path}
\end{figure}

\subsection{Sample-path generation}
The sample path of the process $X_t$ in equation \eqref{switchingsde} can be generated by the Euler approximation \eqref{EMapp} in Appendix:
\begin{equation}
	X_{t_{j}} \approx X_{t_{j-1}} +\lambda(b(\alpha_{t_{j-1}})-X_{t_{j-1}})h +  \Delta_j Z,
\end{equation}
where $ \Delta_jZ \sim \operatorname{NIG}(a,0,\delta h,0)$.
There are two points to be mentioned for sample-path generation:
\begin{itemize}
	\item To estimate the parameters in all states, the sample path must range over the entire state space, which requires a sufficiently large terminal time $T$; 
	\item To ensure accuracy, the step size $h$ must be small enough to satisfy the validity condition of the Cauchy approximation in the Euler approximation.
\end{itemize}
Below, we describe the method used to simulate a high-frequency sample $\{X_{t_j}\}_{j=0}^n$ with $t_j = jh$; recall that $h = T/n$ where $T$ denotes the terminal sampling time.
\begin{itemize}
    \item First, we use the R package \texttt{spuRs} to generate Markov chain $\alpha$ (see \cite{spuRs_book} for details) with a smaller step size $\delta = h/10$. 
    \item Using the data $(\alpha_{\del l})_{l=0}^{10 n}$, we then apply the Euler scheme to generate the data for $X$:
\begin{equation*}
    X_{s_l} = X_{s_{l-1}} +\lambda(b(\alpha_{s_{l-1}})-X_{s_{l-1}})\delta +  Z_{s_l} - Z_{s_{l-1}},
\end{equation*}
where $s_l = l\delta$ 
for $l\in \{0,1,...,10n\}$; this internally generates $\{X_{s_l}\}_{l=0}^{10n}$. 
    \item Next, we select a subsequence of the data $\{X_{t_j}\}_{j=0}^{n}$ from $\{X_{t_l}\}_{l=0}^{10n}$ with step size $h$; by thinning the sequence like this, we obtain high-frequency data $\{X_{t_j}\}_{j=0}^{n}$.
\end{itemize}

Figure \ref{sample path} shows the sample paths of $X$ given initial values of $X_{t_0} = X_0 = 0$, $\alpha_{t_0}=\alpha_0=6$ and a step size of $h=0.1$ 
for terminal times $T$ of 500, 1000, 2000, and 5000.

\subsection{Termination conditions}

To determine when to terminate the EM algorithm, specific termination conditions need to be defined. As the convergence speed of the EM algorithm is generally slower than other numerical methods, these conditions should be carefully chosen to ensure the accuracy of the final estimation. The commonly used termination conditions are the difference of estimations, the relative difference of likelihood function values, and the relative difference of estimations. We also need a small positive value $\epsilon$ as an admissible tolerance for convergence. Now suppose we are at $m+1$ step, then the relative difference of likelihood function values is given by
\begin{align}
	 \mathbb{D}^{1}_{m,m+1} := \frac{|\mbbh_n(\hat{\theta}^{(m+1)};\hat{\theta}^{(m)})-\mbbh_n(\hat{\theta}^{(m)};\hat{\theta}^{(m)})|}{|\mbbh_n(\hat{\theta}^{(m)};\hat{\theta}^{(m)})|},
\end{align}
the relative difference of estimations is given by
\begin{align}
	 \mathbb{D}^{2}_{m,m+1} := \frac{|\hat{\theta}^{(m+1)}-\hat{\theta}^{(m)}|}{|\hat{\theta}^{(m)}|},
\end{align} and the difference of estimations is given by
\begin{align}
	 \mathbb{D}^{3}_{m,m+1} := |\hat{\theta}^{(m+1)}-\hat{\theta}^{(m)}|.
\end{align}
 We choose one of these three differences $\mathbb{D}_{m,m+1}$.
 If the difference $\mathbb{D}_{m,m+1}$ is greater than $\epsilon$, we continue to do the iteration step $m+2$. When the difference $\mathbb{D}_{m,m+1}$ is less than $\epsilon$, we stop the iteration and let $\hat{\theta}^{(m+1)}$ be the final estimates of parameter $\theta$.

\subsection{Estimation}

We consider the estimation of parameter $\theta = (b(1),b(2),\lambda,\delta)$. The stopping criterion for the iteration procedure in Algorithm \ref{AL:EM} is defined by a small value of $\epsilon$, such that the iteration stops when the difference $|\hat{\theta}^{(m+1)} - \hat{\theta}^{(m)}|$ is less than $\epsilon$. Additionally, a suitable positive value of $\rho$ must be chosen to govern the iteration process. The quadratic error used in the simulation is defined by $	|\gamma-\hat{\gamma}|^2$ where $\gamma$ is the true parameter value and $\hat{\gamma}$ is the final estimation of the EM algorithm.  We stop at 300 iterations if the difference can not reach $\epsilon$.

To initialize Algorithms \ref{AL:EM} and \ref{AL:prob}, we require initial parameter values $\hat{\theta}^{(0)}$ and the initial state probabilities $\mathbb{P}\left(\alpha_{t_{0}}=i |\mathbb{X}_{0,0};\hat{\theta}^{(0)}\right)$ for all $i \in \mathcal{S}$. In simulations, we set $\mathbb{P}\left(\alpha_{t_{0}}=i |\mathbb{X}_{0,0};\hat{\theta}^{(0)}\right)$ to be equal for each $i$ and choose $\hat{\theta}^{(0)}$ randomly, with $b(1),b(2),\lambda \in [0,10]$, $\delta \in (0,5]$.

Table \ref{t500_0.02} presents the simulation results for $ T = 500 $, $ h = 0.1 $, and $ n = 5000 $, with $ \epsilon $ set to $ 0.02 $ and $ \rho $ set to $ 0.0001 $. The estimated values and quadratic errors are reported. The total time taken by the algorithm, from initialization to final estimation, was 55.64 seconds on an Apple M1 CPU. Figure \ref{f500_0.02} illustrates the evolution of values during the iteration process. A small value of $ \rho $ was chosen to control the convergence of the algorithm and obtain an appropriate estimate.

Table \ref{t1000_0.05} provides simulation results for $ T = 1000 $, $ h = 0.1 $, and $ n = 10000 $, with $ \epsilon $ set to $ 0.05 $ and $ \rho $ set to $ 0.0001 $. The algorithm completed in 36.38 seconds on an Apple M1 CPU. Figure \ref{f1000_0.05} shows the evolution of values in this iteration process.

For $ T = 2000 $, $ h = 0.1 $, and $ n = 20000 $, Table \ref{t2000_0.1} displays the simulation results, with $ \epsilon $ set to $ 0.1 $ and $ \rho $ to $ 0.0001 $. This iteration process took approximately 1.74 minutes on an Apple M1 CPU, as illustrated in Figure \ref{f2000_0.1}.

Finally, Table \ref{t5000_0.03} presents results for $ T = 5000 $, $ h = 0.1 $, and $ n = 50000 $, with $ \epsilon $ set to $ 0.03 $ and $ \rho $ to $ 0.00001 $. The iteration required 2.97 minutes to complete on an Apple M1 CPU. Figure \ref{f5000_0.03} depicts the evolution of values throughout the iteration process.

\begin{rem}
    Selecting an appropriate value for the value  $\rho$  is crucial in the algorithm. A well-chosen $\rho$ can accelerate convergence and reduce the required iteration. 
    An excessively large $\rho$ may lead to divergence or oscillatory behavior, while an overly small $\rho$ can result in slow convergence, potentially causing the algorithm to terminate prematurely before reaching the optimal solution.
\end{rem}

From the simulation results, the following observations can be made:
\begin{itemize}
    \item In Table \ref{t500_0.02}, the parameters $b(1)$ and $b(2)$ associated with the state of the Markov chain are estimated with high accuracy, while the estimators for $\lambda$ and $\delta$ exhibit lower precision. This is further illustrated in Figure \ref{f500_0.02}, where the estimated value of $\delta$ appears lower than its true value. Moreover, the red curve remains flat over a substantial number of iteration steps (up to step 300), indicating that the deviation in estimations exceeds the threshold $\epsilon = 0.02$. Also, the estimate for $\delta$ consistently falls short of the true value over these iterations. 
    \item As demonstrated in Table \ref{t1000_0.05}, the parameters $b(1)$ and $b(2)$ are estimated accurately, whereas $\lambda$ and $\delta$ continue to be comparatively less precise. Figure \ref{f1000_0.05} shows that all parameter estimates ultimately converge to their true values.
    \item Table \ref{t2000_0.1} shows that  $b(1)$ and $b(2)$ maintain high estimation accuracy, while $\lambda$ and $\delta$ remain somewhat less accurate. In Figure \ref{f2000_0.1}, aside from $b(1)$, the other parameter estimates exhibit minor oscillations around their true values, which suggests a slight perturbation in the estimation process.
    \item In Table \ref{t5000_0.03}, while $b(1)$ and $b(2)$  retain their accuracy, the estimates for $\lambda$ and $\delta$ continue to be less precise by comparison. However, Figure \ref{f5000_0.03} indicates that despite the reduced accuracy for $\lambda$ and $\delta$, these estimates demonstrate stable convergence toward their true values, free of significant oscillations.
\end{itemize}
From the above observations, it is evident that our method provides accurate estimates for parameters in the drift coefficient across all time schemes while, for the noise parameter $\delta$, the estimates exhibit relatively lower precision. As $T$ increases, the estimates for each parameter exhibit greater stability, with minimal variation, underscoring the higher reliability of the estimation method over longer time horizons.

\begin{table}[]
	\begin{tabular}{@{} l l l l @{}}
		\toprule
		& True value & Estimation $T=500$ & Quadratic error \\ \midrule
		$b(1)$&  6& 6.0126082 &  0.00015\\ \midrule
		$b(2)$&  3&  2.9044452&0.00913  \\ \midrule
		$\lambda$&  2& 1.8258804 &0.03031  \\ \midrule
		$\delta$&  1&  0.6516502 &0.12134  \\ \bottomrule
	\end{tabular}
\caption{Estimation of $\theta$ for $T = 500$, $h=0.1$ with $\epsilon=0.02$.}
\label{t500_0.02}
\end{table}

\begin{figure}
	\centering
	\includegraphics[scale=0.6]{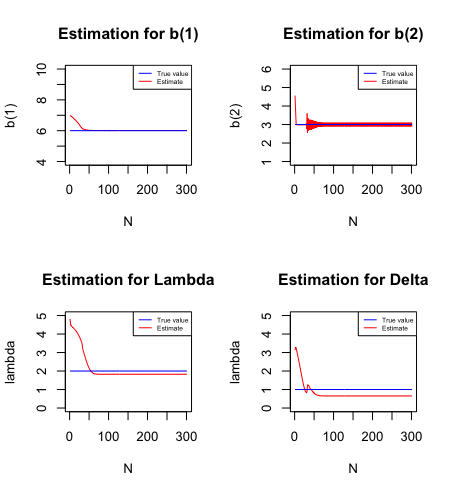}
	\caption{Iteration with $T = 500$, $h=0.1$, $\epsilon=0.02$. The blue line represents the true parameter value, and the red line represents the estimated value.}
	\label{f500_0.02}
\end{figure}

\begin{table}[]
	\begin{tabular}{@{} l l l l @{}}
		\toprule
		& True value & Estimation $T=1000$ & Quadratic error \\ \midrule
		$b(1)$&  6&  6.0130294& 0.00017 \\ \midrule
		$b(2)$&  3& 3.0073506 & 0.00005 \\ \midrule
		$\lambda$&  2& 1.6649264 & 0.11227 \\ \midrule
		$\delta$&  1&0.6390311  & 0.13029 \\ \bottomrule
	\end{tabular}
	\caption{Estimation of $\theta$ for $T = 1000$, $h=0.1$  with $\epsilon=0.05$.}
	\label{t1000_0.05}
\end{table}

\begin{figure}
	\centering
	\includegraphics[scale=0.6]{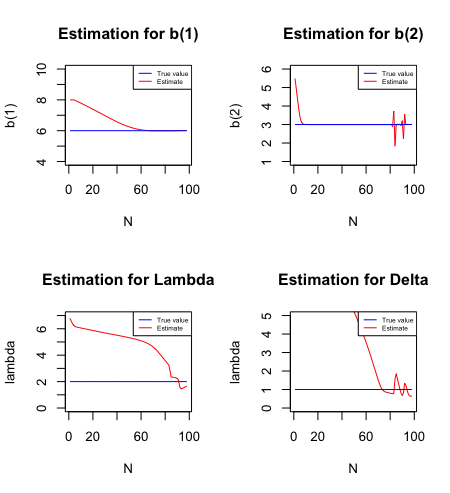}
	\caption{Iteration with $T = 1000$, $h=0.1$, $\epsilon=0.05$. The blue line represents the true parameter value, and the red line represents the estimated value.}
	\label{f1000_0.05}
\end{figure}

\begin{table}[]
	\begin{tabular}{@{} l l l l @{}}
		\toprule
		& True value & Estimation $T=2000$ & Quadratic error \\ \midrule
		$b(1)$&  6&  6.001277& 0.000001 \\ \midrule
		$b(2)$&  3&  2.992654&0.00005  \\ \midrule
		$\lambda$&  2&  1.733450&0.07105  \\ \midrule
		$\delta$&  1&  0.639503&0.12996  \\ \bottomrule
	\end{tabular}
	\caption{Estimation of $\theta$ for $T = 2000$, $h=0.1$  with $\epsilon=0.1$.}
		\label{t2000_0.1}
\end{table}

\begin{figure}
	\centering
	\includegraphics[scale=0.6]{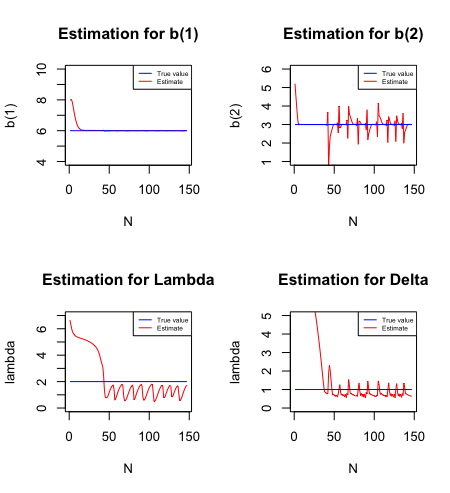}
	\caption{Iteration with $T = 2000$, $h=0.1$, $\epsilon=0.1$. The blue line represents the true parameter value, and the red line represents the estimated value.}
	\label{f2000_0.1}
\end{figure}

\begin{table}[]
	\begin{tabular}{@{} l l l l @{}}
		\toprule
		& True value & Estimation $T=5000$ & Quadratic error \\ \midrule
		$b(1)$&  6& 6.0027876 & 0.000008 \\ \midrule
		$b(2)$&  3& 2.9989096 & 0.000001 \\ \midrule
		$\lambda$&  2& 2.4268411 & 0.18219 \\ \midrule
		$\delta$&  1& 0.6528876 & 0.12049 \\ \bottomrule
	\end{tabular}
	\caption{Estimation of $\theta$ for $T = 5000$, $h=0.1$, $\epsilon=0.03$.}
		\label{t5000_0.03}
\end{table}

\begin{figure}
	\centering
	\includegraphics[scale=0.6]{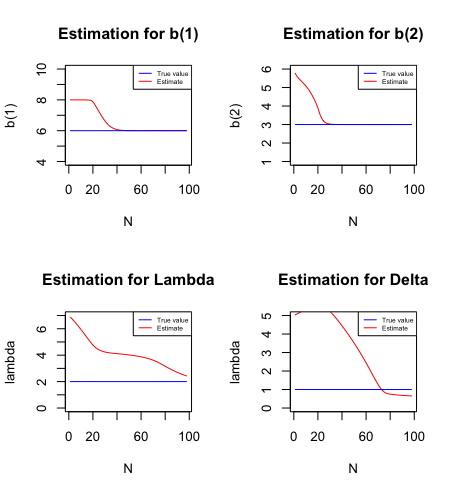}
	\caption{Iteration with $T = 5000$, $h=0.1$, $\epsilon=0.03$. The blue line represents the true parameter value, and the red line represents the estimated value.}
	\label{f5000_0.03}
\end{figure}

\section{Conclusion}

In this paper, we develop a parameter estimation method for regime-switching stochastic differential equations driven by NIG noise. The proposed algorithm is designed for discrete, high-frequency observations and can handle situations where the underlying Markov chain is not directly observed. By combining the EM algorithm with a Cauchy quasi-likelihood approximation, and exploiting the small-time asymptotics of NIG noise under high-frequency sampling, our framework enables the joint estimation of parameters associated with both the SDE coefficients and the noise distribution.
This approach represents an advance over previous methods, such as \cite{ChevallierGoutte2017}, by providing a unified treatment of hidden regimes and non-Gaussian noise in SDEs. We also presented a computational procedure for calculating the necessary conditional probabilities, including those involving the unobserved Markov chain states.

Our simulation results indicate that parameters directly linked to the Markov chain, such as those appearing in the drift and diffusion coefficients for each regime, are estimated with high accuracy. Although the estimation of the NIG noise scale parameter exhibits somewhat lower precision, the overall method still delivers reasonable estimates that improve as the sampling horizon or iteration count increases.
These findings suggest that the EM algorithm, enhanced by a Cauchy quasi-likelihood approximation, is a practical tool for analyzing complex stochastic systems involving regime-switching and non-Gaussian noise. Future research could extend this methodology to more general Lévy-driven SDEs and jointly estimate the transition intensities of the underlying Markov chain.





\subsection*{Code availability}
The code developed for the current study is available from the corresponding author 
on request.

\bigskip

\subsection*{Acknowledgement}
This work was partially supported by WISE program (MEXT) at Kyushu University (YC), and by JST CREST Grant Number JPMJCR2115, Japan, and by JSPS KAKENHI Grant Number 22H01139 (HM).


\section*{Declarations}

\subsection*{Conflict of interest} The authors declare that they have no conflict of interest.

\subsection*{Consent for publication} The authors consent.

\subsection*{Ethical approval} Not applicable.


\section*{Appendix}

\subsection{Second-order derivatives}
\label{sec:2nd.order.derivatives}

When applying \eqref{newton} in the EM algorithm, it becomes necessary to compute the second-order derivatives explicitly. Below, we outline the computation of these second-order derivatives
$\p^2_{\theta}\mbbh_n(\theta;\hat{\theta}^{(m)})$. We have 
\begin{align*}
	\p^2_{\delta} \mbbh_n(\theta;\hat{\theta}^{(m)}) = \sumj \sum_{i \in \mathcal{S}} \sum_{k \in \mathcal{S}}
	&\left(K_{1,j}(\theta)^2K_{2,j}(\theta)^2+K_{2,j}(\theta)K_{5,j}(\theta) \right)
	\\
	&\cdot \mathbb{P}\left(\alpha_{t_{j-1}}=i,\alpha_{t_j}=k| \mathbb{X}_{0,n};\hat{\theta}^{(m)}\right)
\end{align*}
where
\begin{align*}
	K_{5,j}(\theta) = \frac{-2\pi(X_{t_j}-\mu_{j-1}(\lambda))^2}{\delta^3 h}
\end{align*}
and
\begin{align*}
	\p^2_{\lambda} \mbbh_n(\theta;\hat{\theta}^{(m)}) = \sumj \sum_{i \in \mathcal{S}}\sum_{k \in \mathcal{S}}
	&\left(K_{2,j}(\theta)^2K_{3,j}(\theta)^2 +	K_{2,j}(\theta)	K_{6,j}(\theta) \right)
	\\
	& \cdot \mathbb{P}\left(\alpha_{t_{j-1}}=i,\alpha_{t_j}=k| \mathbb{X}_{0,n};\hat{\theta}^{(m)}\right)
\end{align*}
where
\begin{align*}
	K_{6,j}(\theta) = \frac{-2\pi h(b(i)-X_{t_{j-1}})^2}{ \delta}
\end{align*}
and
\begin{align*}
	\p_{\lambda}\p_{\delta} \mbbh_n(\theta;\hat{\theta}^{(m)}) = \sumj \sum_{i \in \mathcal{S}}\sum_{k \in \mathcal{S}}
	&\left(K_{2,j}(\theta)^2K_{3,j}(\theta)K_{1,j}(\theta) +K_{2,j}(\theta)K_{7,j}(\theta)  \right)
	\\
	& \cdot \mathbb{P}\left(\alpha_{t_{j-1}}=i,\alpha_{t_j}=k| \mathbb{X}_{0,n};\hat{\theta}^{(m)}\right)
\end{align*}
where
\begin{align*}
	K_{7,j}(\theta) = \frac{-2\pi(X_{t_j}-\mu_{j-1}(\lambda))(b(i)-X_{t_{j-1}})}{ \delta^2 h}
\end{align*}
and
\begin{align*}
	\p_{b(l)}\p_{\delta} \mbbh_n(\theta;\hat{\theta}^{(m)}) = \sumj \sum_{k \in \mathcal{S}}
	 &\left(K_{2,l,j}^{\star}(\theta)^2K_{1,j}(\theta)K_{4,l,j}(\theta) + K_{2,l,j}^{\star}(\theta)K_{8,j}(\theta)  \right)
	\\
	& \cdot  \mathbb{P}\left(\alpha_{t_{j-1}}=l,\alpha_{t_j}=k| \mathbb{X}_{0,n};\hat{\theta}^{(m)}\right)
\end{align*}
where
\begin{align*}
	K_{8,j}(\theta) := \frac{-2\pi(X_{t_j}-(X_{t_{j-1}} +\lambda(b(l)-X_{t_{j-1}})h))\lambda}{ \delta^2}.
\end{align*}
and
\begin{align*}
	\p_{b(l)}\p_{\lambda} \mbbh_n(\theta;\hat{\theta}^{(m)}) = \sumj \sum_{k \in \mathcal{S}}
	 &\left(K_{2,l,j}^{\star}(\theta)^2K^{\star}_{3,j}(\theta)K_{4,l,j}(\theta) + K_{2,l,j}^{\star}(\theta)K_{9,j}(\theta)  \right)
	\\
	& \cdot  \mathbb{P}\left(\alpha_{t_{j-1}}=l,\alpha_{t_j}=k| \mathbb{X}_{0,n};\hat{\theta}^{(m)}\right)
\end{align*}
where
\begin{align*}
	&K^{\star}_{3,j}(\theta) := \frac{2\pi(X_{t_j}-\mu_{j-1}(\lambda))(b(l)-X_{t_{j-1}})}{ \delta},
	\\
	&K_{9,j}(\theta) := \frac{2\pi(X_{t_j}-(X_{t_{j-1}} +2\lambda(b(l)-X_{t_{j-1}})h))}{ \delta}
\end{align*}
and we have
\begin{align*}
	\p^2_{b(l)} \mbbh_n(\theta;\hat{\theta}^{(m)}) = \sumj \sum_{k \in \mathcal{S}} &\left(K_{2,l,j}^{\star}(\theta)^2K_{4,l,j}(\theta)^2 + K_{2,l,j}^{\star}(\theta)\frac{-2\pi h \lambda^2}{\delta}  \right)
	\\
	& \cdot \mathbb{P}\left(\alpha_{t_{j-1}}=l,\alpha_{t_j}=k| \mathbb{X}_{0,n};\hat{\theta}^{(m)}\right)
\end{align*}
and when $k \neq l$
\begin{align*}
	\p_{b(l)}\p_{b(k)} \mbbh_n(\theta;\hat{\theta}^{(m)}) = 0.
\end{align*}
Also, we have 
\begin{align*}
		\p_{q_{ll}}^2 \mbbh_n(\theta;\hat{\theta}^{(m)}) &= \sumj  -\frac{h^2}{(1+ q_{ll}h)^2}\mathbb{P}\left(\alpha_{t_{j-1}}=l,\alpha_{t_j}=l| \mathbb{X}_{0,n};\hat{\theta}^{(m)}\right),
	\\
	\p_{q_{lm}}^2 \mbbh_n(\theta;\hat{\theta}^{(m)}) &= \sumj  -\frac{1}{q_{lm}^2}\mathbb{P}\left(\alpha_{t_{j-1}}=l,\alpha_{t_j}=m| \mathbb{X}_{0,n};\hat{\theta}^{(m)}\right), \,\,\, \operatorname{for} l \neq m.
\end{align*}
and
\begin{align*}
	&\p_{q_{lm}}\p_{\delta} \mbbh_n(\theta;\hat{\theta}^{(m)}) = 0
	\\
	&\p_{q_{lm}}\p_{\lambda} \mbbh_n(\theta;\hat{\theta}^{(m)}) = 0
	\\
	&\p_{q_{lm}}\p_{b(i)} \mbbh_n(\theta;\hat{\theta}^{(m)}) = 0,
	  \,\,\, \operatorname{for \, all} \, l, \, i, \, m.
\end{align*}

\subsection{Numerical approximation}
\label{sec:A1}

Let $w$ be $r$-dimensional standard Wiener process and $N$ an independent Poisson random measure on $\mbbr^{+}\times (\mbbr^d -{0})$. For each $t >0$ and $A \in \mathcal{B}(\mbbr^d -0)$, the compensator of $N$ is given by $\tilde{N}(t,A) = N(t,A)-t\nu(A)$ where $\nu$ is the L\'{e}vy measure. In this section, we consider the convergence of the Euler approximation of the following switching stochastic differential equation with L\'{e}vy noise
\begin{align}
dX_t =& f(X_t,\alpha_t)dt + \sigma(X_t,\alpha_t)dw_t
\nn \\
&+ \int_{|y|<c}F(X_{t-},y,\alpha_{t-})\cpoim{t} +  \int_{|y|\geq c}G(X_{t-},y,\alpha_{t-})\poim{t},
\label{slevysde}
\end{align} where $c \in (0,\infty)$ and 
\begin{align*}
	\int_{|y|<c}F(X_{t-},y,\alpha_{t-})\cpoim{t} := L^2-\lim_{n \to \infty}\int_{\epsilon_n<|y|<c}F(X_{t-},y,\alpha_{t-})\cpoim{t},
\end{align*}where $\epsilon_n$ decreases monotonically to 0.
This L\'{e}vy-Ito decomposition formulation of the stochastic differential equation driven by L\'{e}vy noise with no switching can be found in \cite{applebaum2009levy}. Here we always assume Markov chain $\alpha_t$ with finite state $\mathcal{S} = \{1,2,...,N\}$ is independent of $w_t$ and $N(t,A)$. 

We define the Euler approximation to the equation \eqref{slevysde}.  Our construction of the Euler approximation is based on Chapter 4 in \cite{mao2006stochastic} for switching diffusion. We set $h>0$ and $t_j = jh$ for $j \in \mathbb{N}$.
Let 
\begin{align*}
	\alpha^h(t) =\alpha^h_j  \quad  \text{for} \quad  t \in [t_j,t_{j+1}), \quad \alpha^h_j =\alpha_{jh}, \quad \alpha^h(0) = \alpha_0.
\end{align*}
Now define $X^h_j$ for $j \geq 0$ by $ X^h_0 = X_0$, and 
\begin{align*}
	X^h_j  = &X^h_{j-1} + f(X^h_{j-1},\alpha^h_{j-1})h + \sigma(X^h_{j-1},\alpha^h_{j-1})\Delta_j w
	\notag\\& +\int_{t_{j-1} }^{t_j}\int_{|y|<c}F(X^h_{j-1},y,\alpha^h_{j-1})\tilde{N}(dt,dy)
	\notag\\ &+ \int_{t_{j-1}}^{t_j} \int_{|y|\geq c}G(X^h_{j-1},y,\alpha^h_{j-1})N(dt,dy),
\end{align*} where $\Delta_j w := w_{t_j}-w_{t_{j-1}}$. Let 
\begin{equation*}
	X^h(t) = X^h_j \,\, \text{for} \,\, t \in [t_j,t_{j+1}).
\end{equation*}
Then we define the Euler approximation to be 
\begin{align}
	\bar{X}_t = X_0 &+ \int_{0}^{t}f(X^h(s),\alpha^h(s))ds + \int_{0}^{t}\sigma(X^h(s),\alpha^h(s))dw_s \notag
	\\&+\int_{0}^{t} \int_{|y|<c}F(X^h(s),y,\alpha^h(s))\cpoim{s}  \notag
	\\&+  \int_{0}^{t}\int_{|y|\geq c}G(X^h(s),y,\alpha^h(s))\poim{s}.
	\label{EMapp}
\end{align}
Note that if $t$ lies in some interval $[t_j,t_{j+1})$, then
\begin{align*}
	\bar{X}_t = X_{jh} &+ \int_{t_j}^{t}f(X_{jh},\alpha_{jh})ds + \int_{t_j}^{t}\sigma(X_{jh},\alpha_{jh})dw_s \notag
	\\&+\int_{t_j}^{t} \int_{|y|<c}F(X_{jh},y,\alpha_{jh})\cpoim{s}  \notag
	\\&+  \int_{t_j}^{t}\int_{|y|\geq c}G(X_{jh},y,\alpha_{jh})\poim{s}.
\end{align*}

Let $a(x,y,i,j):=\sigma(x,i)\sigma(y,j)^T$ and $\norm{a}=\sum_{k=1}|a_{kk}|$.
Now we impose some assumptions to ensure the existence and uniqueness of the solution in \eqref{slevysde} and for further usage.
\begin{enumerate}
	\item[(A1)] Lipschitz condition on $f,a,F,G$. {}
	There exists $C>0$ such that for all $x_1,x_2 \in \mbbr^d, i \in \mathcal{S}$,
	\begin{align*}
		&|f(x_1,i)-f(x_2,i)|^2 + \norm{a(x_1,x_2,i,i)-2a(x_1,x_2,i,i) +a(x_2,x_2,i,i)} 
		\\ &+  \int_{|y|<c}|F(x_1,y,i)-F(x_2,y,i)|^2\nu(dy) + \int_{|y|\geq c} |G(x_1,y,i)-G(x_2,y,i)|^2\nu(dy) \leq C|x_1- x_2|^2
	\end{align*}
   \item[(A2)] Growth condition on $f,a,F,G$.{}
   There exists $C>0$ such that for all $x \in \mbbr^d, i \in \mathcal{S}$,
   \begin{align*}
   	&|f(x,i)|^2 + \norm{a(x,x,i,i)} 
   	\\ &+  \int_{|y|<c}|F(x_y,i)|^2\nu(dy) + \int_{|y|\geq c} |G(x,y,i)|^2\nu(dy) \leq C(1+|x|^2)
   \end{align*}
 \item[(A3)] $x \mapsto G(x,y,i)$ is continuous for all $y \geq c, i \in \mathcal{S}$.
 \item[(A4)] $\E(|X_0|^2) \leq \infty$.
\end{enumerate}
Under assumptions (A1), (A2), and (A3), the uniqueness and existence of the solution to equation \eqref{slevysde} can be shown by a stopping time argument, see Chapter 3 in \cite{mao2006stochastic}.

Before we start proving the convergence of the Euler approximation, we prove the following Lemma should first, which is about the boundedness of the second-order moment for the approximation \eqref{EMapp}.
\begin{lem}
	Under the assumptions (A1),(A2),(A3),(A4), there exist a constant $C$ independent of $h$ such that Euler approximation \eqref{EMapp} satisfies
	\begin{equation}
		\E \left( \displaystyle\sup_{0\leq t \leq T} |\bar{X}_t|^2\right) \leq C.
		\label{LemEM1}
	\end{equation}
\end{lem}

\begin{proof} Note that
	\begin{align*}
		|\bar{X}_t|^2 \leq &C \Big(|X_0|^2 + \left|\int_{0}^{t}f(X^h(s),\alpha^h(s))ds\right|^2 + \left|\int_{0}^{t}\sigma(X^h(s),\alpha^h(s))dw_s\right|^2
		\\ &+\left|\int_{0}^{t} \int_{|y|<c}F(X^h(s),y,\alpha^h(s))\cpoim{s}\right|^2
		\\&+\left|\int_{0}^{t}\int_{|y|\geq c}G(X^h(s),y,\alpha^h(s))\poim{s}\right|^2\Big).
	\end{align*}
Through the definition of the compensator $\tilde{N}$,  we have 
\begin{align*}
	\int_{0}^{t}\int_{|y|\geq c}G(X^h(s),y,\alpha^h(s))\poim{s}= &\int_{0}^{t}\int_{|y|\geq c}G(X^h(s),y,\alpha^h(s))\cpoim{s}
	\\&+\int_{0}^{t}\int_{|y|\geq c}G(X^h(s),y,\alpha^h(s))\nu(dy)ds.
\end{align*}
Now, the first term on the right-hand side is a martingale, and the second term on the right-hand side is an integral over a set with finite measure, since $\nu(\{|y|\geq c\}) \leq \infty$; see Section 4.2 in \cite{applebaum2009levy}. Therefore, by Kunita's first inequality (Section 4.4 in \cite{applebaum2009levy}) and Jensen's inequality, we obtain that
\begin{align}
	\E \left( \displaystyle\sup_{0\leq t \leq T} \left|\int_{0}^{t}\int_{|y|\geq c}G(X^h(s),y,\alpha^h(s))\poim{s}\right|^2\right) &\leq C  \E\left(\int_{0}^{T}\int_{|y|\geq c}\left|G(X^h(s),y,\alpha^h(s))\right|^2\nu(dy)ds\right) \notag
	\\&\leq C \int_{0}^{T}\E(1+|X^h(s)|^2)ds,
	\label{poiterm}
\end{align} where the last step is due to the growth condition.
Now a similar procedure can be applied to the compensated Poisson-integral term. By Kunita's inequality and the growth condition, it follows that
\begin{equation}
	\E \left( \displaystyle\sup_{0\leq t \leq T} \left|\int_{0}^{t} \int_{|y|<c}F(X^h(s),y,\alpha^h(s))\cpoim{s}\right|^2\right) \leq C \int_{0}^{T}\E(1+|X^h(s)|^2)ds.
	\label{cpoiterm}
\end{equation}
Also by Jensen's inequality, we have
\begin{equation}
	\E \left( \displaystyle\sup_{0\leq t \leq T} \left|\int_{0}^{t}f(X^h(s),\alpha^h(s))ds\right|^2\right) \leq C \int_{0}^{T}\E(1+|X^h(s)|^2)ds.
	\label{rieterm}
\end{equation}
By Buckholder-Davis-Gundy's inequality, we can obtain a similar result for the Wiener-integral term 
\begin{equation}
		\E \left( \displaystyle\sup_{0\leq t \leq T} \left|\int_{0}^{t}\sigma(X^h(s),\alpha^h(s))dw_s\right|^2\right) \leq C \int_{0}^{T}\E(1+|X^h(s)|^2)ds.
	\label{bmterm}
\end{equation}
Now substituting \eqref{poiterm}, \eqref{cpoiterm}, \eqref{rieterm}, \eqref{bmterm} into the first inequality of $|\bar{X}|^2$ we have
\begin{align*}
	\E \left( \displaystyle\sup_{0\leq t \leq T} |\bar{X}_t|^2\right) &\leq C \int_{0}^{T}\E(1+|X^h(s)|^2)ds
	\\&\leq C + C \int_{0}^{T}\E(\displaystyle\sup_{0\leq t \leq r}|\bar{X}_r|^2)ds.
\end{align*} Then Gronwall's inequality yields the desired result
\begin{equation*}
	\E \left( \displaystyle\sup_{0\leq t \leq T} |\bar{X}_t|^2\right) \leq C.
\end{equation*}
\end{proof}

Having Lemma \ref{LemEM1} in hand, we can now state the main convergence theorem of the Euler approximation
\begin{thm}
	Under the assumptions (A1),(A2),(A3),(A4), there exists a constant $C$ independent of $h$ such that 
	\begin{equation}
		\E \left( \displaystyle\sup_{0\leq t \leq T} |\bar{X}_t-X_t|^2\right) \leq Ch.
		\label{emconv}
	\end{equation} 
\end{thm}
\begin{proof}
Recalling the expressions of $X_t$ and $\bar{X}_t$, we can apply Jensen's inequality and Kunita's first inequality to obtain that 
\begin{align}
		\E \left( \displaystyle\sup_{0\leq t \leq T} |\bar{X}_t-X_t|^2\right) \leq&  C \E \left(\int_{0}^{T} |f(X_s,\alpha_s)-f(X^h(s),\alpha^h(s))|^2ds\right)
		\notag \\
		&+ C \E \left(\int_{0}^{T} \norm{\sigma(X_s,\alpha_s)-\sigma(X^h(s),\alpha^h(s))}ds\right) 
		\notag \\
		& + C\E \left(\int_{0}^{T}\int_{|y|<c} |F(X_s,y,\alpha_s)-F(X^h(s),y,\alpha^h(s))|^2\nu(dy)ds\right)
		\notag \\
		& + C\E \left(\int_{0}^{T}\int_{|y|\geq c} |G(X_s,y,\alpha_s)-G(X^h(s),y,\alpha^h(s))|^2\nu(dy)ds\right)
		\label{ThmEuler1}.
\end{align}
Now we consider the first term in the right-hand side of \eqref{ThmEuler1}. Note that 
\begin{equation*}
	f(X_s,\alpha_s)-f(X^h(s),\alpha^h(s)) = f(X_s,\alpha_s)-f(X^h(s),\alpha_s) + f(X^h(s),\alpha_s)-f(X^h(s),\alpha^h(s)).
\end{equation*}
We have 
\begin{align*}
	\E \left(\int_{0}^{T} |f(X_s,\alpha_s)-f(X^h(s),\alpha^h(s))|^2ds\right) \leq& 2 \E \left(\int_{0}^{T} |f(X_s,\alpha_s)-f(X^h(s),\alpha_s)|^2ds\right) \\
	&+ 2\E \left(\int_{0}^{T} |f(X^h(s),\alpha_s)-f(X^h(s),\alpha^h(s))|^2ds\right)
	\\ \leq& C \E \left(\int_{0}^{T} |X_s-X^h(s)|^2ds\right)
	\\ &+  2\E \left(\int_{0}^{T} |f(X^h(s),\alpha_s)-f(X^h(s),\alpha^h(s))|^2ds\right),
\end{align*} where the last step is followed by the Lipschitz condition. Now we have a difference $ |f(X^h(s),\alpha_s)-f(X^h(s),\alpha^h(s))|^2$ which contains two different Markov processes. To estimate this term, we take a partition on an interval $[0,T]$. Let $n= [\frac{T}{h}]$, integer part of $\frac{T}{n}$, and the partition is taken by 
\begin{equation*}
	0=t_0 < t_1 <, ... , < t_n<t_{n+1} =T; \,\, t_j=jh \,\, \text{for} \,\, j= 0,1,...,n. 
\end{equation*}
We can obtain that
\begin{align*}
	\E \left(\int_{0}^{T} |f(X^h(s),\alpha_s)-f(X^h(s),\alpha^h(s))|^2ds\right) &= \sum_{j=0}^{n}\E \left(\int_{t_j}^{t_{j+1}} |f(X^h(s),\alpha_s)-f(X^h(s),\alpha^h(s))|^2ds\right)
	\\ &= \sum_{j=0}^{n}\E \left(\int_{t_j}^{t_{j+1}} |f(X^h(s),\alpha_s)-f(X^h(s),\alpha_{jh})|^2ds\right)
\end{align*} by definition of $\alpha^h(t)$. Now If $\alpha_s=\alpha_{jh}$ for $s \in[t_j,t_{j+1})$, the right-hand side of the above equality will be zero. We consider the case where $\alpha_s \neq \alpha_{jh}$ for some $s \in[t_j,t_{j+1})$. Direct computations yield  
\begin{align*}
	&\E \left(\int_{t_j}^{t_{j+1}} |f(X^h(s),\alpha_s)-f(X^h(s),\alpha_{jh})|^2ds\right)
 \\
 &= \E \left(\int_{t_j}^{t_{j+1}} |f(X^h(s),\alpha_s)-f(X^h(s),\alpha_{jh})|^2 I_{\{\alpha_s=\alpha_{jh}\}} ds\right) 
 \nn\\
    & {}\qquad + \E \left(\int_{t_j}^{t_{j+1}} |f(X^h(s),\alpha_s)-f(X^h(s),\alpha_{jh})|^2 I_{\{\alpha_s \neq \alpha_{jh}\}} ds\right)
	\\ &\leq C\E \left(\int_{t_j}^{t_{j+1}} \left(|f(X^h(s),\alpha_s)|^2+|f(X^h(s),\alpha_{jh})|^2\right) I_{\{\alpha_s \neq \alpha_{jh}\}} ds\right) 
	\\ &\leq C\E \left(\int_{t_j}^{t_{j+1}} \left(1+|X^h(s)|^2\right) I_{\{\alpha_s \neq \alpha_{jh}\}} ds\right).
\end{align*} 
Since $X^h(s)$ is $\mathcal{F}_{t_j}$-measurable, taking the conditional expectation with respect to $\mathcal{F}_{t_j}$ in the last expression above gives
\begin{align*}
	\E \left(\int_{t_j}^{t_{j+1}} \left(1+|X^h(s)|^2\right) I_{\{\alpha_s \neq \alpha_{jh}\}} ds\right) =& \int_{t_j}^{t_{j+1}}\E\left( \E \left(\left(|1+|X^h(s)|^2\right) I_{\{\alpha_s \neq \alpha_{jh}\}}| \mathcal{F}_{t_j}\right) \right)ds
	\\
	=& \int_{t_j}^{t_{j+1}}\E\left( \left(|1+|X^h(s)|^2\right) \E\left( I_{\{\alpha_s \neq \alpha_{jh}\}}| \mathcal{F}_{t_j}\right) \right)ds.
\end{align*}
Then we use transition probability \eqref{CTMCgenerator} to compute the conditional expectation of the indicator function. We have
\begin{align*}
        \E\left( I_{\{\alpha_s \neq \alpha_{jh}\}}| \mathcal{F}_{t_j}\right) =& 
	\E\left( I_{\{\alpha_s \neq \alpha_{jh}\}}| \alpha_{jh}\right) 
    \\
	=& \sum_{i \in \mathcal{S}} \mathbb{P}(\alpha_s \neq i | \alpha_{jh}=i) I_{\{\alpha_{jh}=i\}}
	\\
	=& \sum_{i \in \mathcal{S}} \left( \sum_{i \neq l}\left(q_{il}(s-jh) + o(s-jh)\right) I_{\{\alpha_{jh}=i\}} \right)
	\\
	\leq& C h.
\end{align*}
Applying this inequality with \eqref{LemEM1}, it follows that
\begin{align*}
	\E \left(\int_{t_j}^{t_{j+1}} |f(X^h(s),\alpha_s)-f(X^h(s),\alpha^h(s))|^2ds\right) &\leq\int_{t_j}^{t_{j+1}}\E\left( \left(|1+|X^h(s)|^2\right)(Ch + o(h))\right)ds
	\\&\leq Ch^2,
\end{align*}
and also we have
\begin{equation*}
	\E \left(\int_{0}^{T} |f(X^h(s),\alpha_s)-f(X^h(s),\alpha^h(s))|^2ds\right) \leq Ch.
\end{equation*}
Combining the above estimates we can give an estimate 
\begin{equation}
		\E \left(\int_{0}^{T} |f(X_s,\alpha_s)-f(X^h(s),\alpha^h(s))|^2ds\right) \leq C \E \left(\int_{0}^{T} |X_s-X^h(s)|^2ds\right) + Ch.
\end{equation}
Now we have successfully estimated the first term in \eqref{ThmEuler1}. Other terms in \eqref{ThmEuler1} can be estimated due to a similar argument. We skip detailed computations and just put the results here:
\begin{align}
	&\E \left(\int_{0}^{T} \norm{\sigma(X_s,\alpha_s)-\sigma(X^h(s),\alpha^h(s))}ds\right) \leq C \E \left(\int_{0}^{T} |X_s-X^h(s)|^2ds\right) + Ch,
	\\ &\E \left(\int_{0}^{T}\int_{|y|<c} |F(X_s,y,\alpha_s)-F(X^h(s),y,\alpha^h(s))|^2\nu(dy)ds\right) 
 \\  &{}\qquad \leq C \E \left(\int_{0}^{T} |X_s-X^h(s)|^2ds\right) + Ch,
	\\ &\E \left(\int_{0}^{T}\int_{|y|\geq c} |G(X_s,y,\alpha_s)-G(X^h(s),y,\alpha^h(s))|^2\nu(dy)ds\right) 
 \\
 &{}\qquad \leq C \E \left(\int_{0}^{T} |X_s-X^h(s)|^2ds\right) + Ch.
\end{align}
Substituting above estimates into \eqref{ThmEuler1} we obtain
\begin{equation}
	\E \left( \displaystyle\sup_{0\leq t \leq T} |\bar{X}_t-X_t|^2\right) \leq C \E \left(\int_{0}^{T} |X_s-X^h(s)|^2ds\right) + Ch.
	\label{thmEuler2}
\end{equation}
Note that 
\begin{align}
	\E \left(\int_{0}^{T} |X_s-X^h(s)|^2ds\right)  \leq 2\E \left(\int_{0}^{T} |X_s-\bar{X}_s|^2ds\right) + 2 \E \left(\int_{0}^{T} |\bar{X}_s-X^h(s)|^2ds\right) 
	\label{thmEuler3}
\end{align}
and if $s \in [t_j,t_{j+1})$, by definition of the Euler approximation, Jensen's inequality, and Doob' martingale inequality, we have
\begin{align*}
	\E |\bar{X}_s-X^h(s)|^2 \leq& C \E \left((s-t_j)\int_{t_j}^{s}|f(X_{jh},\alpha_{jh})|^2du\right)
+C\E \left(\int_{t_j}^{s}\norm{\sigma(X_{jh},\alpha_{jh})}du\right)
	\\ &+ C\E \left(\int_{t_j}^{s} \int_{|y|<c}|F(X_{jh},y,\alpha_{jh})|^2\nu(dy)du\right)
	\\ &+ C\E \left(\int_{t_j}^{s} \int_{|y|\geq c}|G(X_{jh},y,\alpha_{jh})|^2\nu(dy)du\right).
\end{align*}
Then the growth condition and \eqref{LemEM1} finally yield
\begin{equation*}
		\E |\bar{X}_s-X^h(s)|^2 \leq C \E \left((1+|X_{jh}|^2)(h^2+h)\right) \leq Ch^2+ Ch \leq Ch.
\end{equation*}
Therefore we obtain that
 \begin{equation*}
 	\E \left(\int_{0}^{T} |\bar{X}_s-X^h(s)|^2ds\right) \leq Ch.
 \end{equation*}
Substituting this and \eqref{thmEuler3} into \eqref{thmEuler2} we have 
\begin{align*}
	\E \left( \displaystyle\sup_{0\leq t \leq T} |\bar{X}_t-X_t|^2\right) \leq& C\E \left(\int_{0}^{T} |X_s-\bar{X}_s|^2ds\right) + Ch
	\\
	\leq& C\int_{0}^{T} \E \left(\displaystyle\sup_{0\leq t \leq r} |X_t-\bar{X}_t|^2\right)dr + Ch.
\end{align*}
Now Grownwall's inequality yields the desired result
\begin{equation*}
		\E \left( \displaystyle\sup_{0\leq t \leq T} |\bar{X}_t-X_t|^2\right) \leq Ch.
\end{equation*}

\end{proof}


\subsection{Markov property}

In this section, we establish that the two-component process $(X_t,\alpha_t)$, which is comprised of the strong solution of equation $\eqref{switchingsde}$ and the continuous-time Markov chain, possesses the Markov property. We prove a general result on the Markov property, which includes equation $\eqref{switchingsde}$ as a special case. The proof is based on the idea behind the proof of the Markov property for L\'{e}vy-driven SDEs without Markov switching, as detailed in \cite{applebaum2009levy}. By applying the Markov property, we can then derive various properties and behaviors of the two-component process, which can be useful in further applications and analyses.

The Markov property of the two-component process depends on the existence and uniqueness of a strong solution to equation \eqref{slevysde}. This is ensured by Theorem 2.1 in \cite{xi2017feller}, assuming that conditions (A1), (A2), and (A3) in Section \ref{sec:A1} hold. The strong solution $X_t$ is unique under these assumptions, as demonstrated in \cite{mao2006stochastic} and \cite{xi2017feller}.
To prove the Markov property, we will first state a theorem, followed by the necessary preliminary background required for the proof. Finally, we will present a detailed proof of the theorem.

\begin{thm}
Assuming that the equation \eqref{slevysde} with initial value $X_0 = x \in \mathbb{R}^d$ and $\alpha_0 = i \in \mathcal{S}$ a.s.  admits a unique solution. Then, the two-component process $(X_t,\alpha_t)$, which is formed by the strong solution $X_t$ of equation \eqref{slevysde} together with its associated Markov chain $\alpha_t$, is a Markov process.
\label{markovproperty}
\end{thm}

To proceed, we need a property of continuous-time Markov chain and a classical Lemma.
We first note that the homogeneous continuous-time Markov chain $\alpha_t$ with generator $Q=(q_{ij})_{N\times N}$ can be represented by a stochastic integral with respect to a Poisson random measure (see \cite{mao2006stochastic}, \cite{skorokhod2009asymptotic} and \cite{yin2010stability}). For $i, j \in \mathcal{S}$  with  $j \neq i$, let  $\Delta_{i j}$  be the consecutive, left closed and right open intervals of the real line, each having length  $q_{i j}$ such that
\begin{align*}
	&\Delta_{12} = [0,q_{12}), \: \Delta_{13} = [q_{12},q_{12}+q_{13}), \: 
	\\
	&... \:  
	\\
	& \Delta_{1N}=\left[\sum_{j=2}^{N-1}q_{1j},\sum_{j=2}^{N}q_{1j}\right),
	\\
	&\Delta_{21}=\left[\sum_{j=2}^{N}q_{1j},\sum_{j=2}^{N}q_{1j}+q_{21}\right), \: \Delta_{23}=\left[\sum_{j=2}^{N}q_{1j}+q_{21},\sum_{j=2}^{N}q_{1j}+q_{21} + q_{23}\right),
	\\
	&...
	\\
	&\Delta_{2N}=\left[\sum_{j=2}^{N}q_{1j}+\sum_{j=1,j\neq2}^{N-1}q_{2j},\sum_{j=2}^{N}q_{1j}+\sum_{j=1,j\neq2}^{N}q_{2j}\right)
\end{align*} and so forth.
Define a function  $h: \mathcal{S} \times \mathbb{R} \mapsto \mathbb{R} $ by
\begin{equation*}
	h(i , z)=\sum_{j \in \mathcal{S}}(j-i) I_{\left\{z \in \Delta_{i j}\right\}} .
\end{equation*}
This says that for each  $i \in \mathcal{S}$ , if $ z \in \Delta_{i j}$ , then $h( i, z)=j-i $; otherwise  $h( i, z)=0 $. Then as in Chapter 1 in \cite{mao2006stochastic},
\begin{equation}
d \alpha_t=\int_{\mathbb{R}} h\left(\alpha_{t-}, z\right) \Pi(d t, d z),
\label{mcequation}
\end{equation}
where  $\Pi(d t, d z) $ is a Poisson random measure with intensity  $d t \times m(d z)$ , and  $m(\cdot) $ is the Lebesgue measure on $ \mathbb{R}$ . The Poisson random measure  $\Pi(\cdot, \cdot)$ here is independent of the Wiener process $w$ and Poisson random measure  $N(\cdot, \cdot)$ in equation \eqref{slevysde}.
Before proceeding to the proof of Theorem \ref{markovproperty}, it is necessary to introduce a classical lemma; see Lemma 1.1.9 in \cite{applebaum2009levy}.
\begin{lem}
	Let $(\Omega,\mathcal{F},\mathbb{P})$ be a probability space and $\mathcal{G}$ be a sub-$\sigma$-algebra of $\mathcal{F}$. If $X$ and $Y$ are random variables such that $X$ is $\mathcal{G}$-measurable and $Y$ is independent with $\mathcal{G}$. Then 
	\begin{equation*}
		\mathbb{E}\left(f(X,Y)|\mathcal{G}\right) = K_f(X), 
	\end{equation*} for any bounded Borel measurable function $f(x,y)$, where $K_f(x) = \mathbb{E}\left(f(x,Y)\right)$ for each $x\in \mathbb{R}^d$. 
\label{mclem}
\end{lem}

Now we give a proof of Theorem \ref{markovproperty}. 
\begin{proof}[Proof of Theorem  \ref{markovproperty} ]
	First define $\mathcal{G}_s = \sigma\{w_{s+u}-w_{s}, N(s+u,A)-N(s,A), \Pi(s+u,B)-\Pi(s,B) : u\geq0, A \in\mathcal{B}(\mathbb{R}^d-0), B \in \mathcal{B}(\mathbb{R})\}$. Then $\mathcal{G}_s$ is independent of $\mathcal{F}_s$, since $w_t$, $N(t,A)$ and $\Pi(t,B)$ are L\'{e}vy processes.
	Let $\Phi_{0,t}(x,i)$ be a solution of equation \eqref{slevysde} with initial value $X_0 = x  $ and $\alpha_0 = i $ a.s. and $\Lambda_{0,t}(i)$ be a solution of equation \eqref{mcequation} with $\alpha_0 = i $ a.s.:
	\begin{align*}
			\Phi_{0,t}(x,i)	 =& \: x + \int_{0}^{t} f(\Phi_{0,u}(x,i),\Lambda_{0,u}(i))du + \int_{0}^{t} \sigma(\Phi_{0,u}(x,i),\Lambda_{0,u}(i))dw_u
			\\
			&+ \int_{0}^{t}\int_{|y|<c}F(\Phi_{0,u-}(x,i),y,\Lambda_{0,u-}(i))\cpoim{u} 
		\\
		&+  \int_{0}^{t}\int_{|y|\geq c}G(\Phi_{0,u-}(x,i),y,\Lambda_{0,u-}(i))\poim{u},
		\\
		\Lambda_{0,t} (i)=&\: i + \int_{0}^{t} \int_{\mathbb{R}} h\left(\Lambda_{0,u-}(i), z\right) \Pi(d u, d z).
	\end{align*}
	Also we have
	\begin{align*}
		\Phi_{s,t}(\Phi_{0,s}(x,i),\Lambda_{0,s} (i)) =& \: \Phi_{0,s}(x,i) + \int_{s}^{t} f(\Phi_{s,u}(\Phi_{0,s}(x,i),\Lambda_{0,s} (i)),\Lambda_{s,u}(\Lambda_{0,s} (i)))du 
		\\
		&+ \int_{s}^{t} \sigma(\Phi_{s,u}(\Phi_{0,s}(x,i),\Lambda_{0,s} (i)),\Lambda_{s,u}(\Lambda_{0,s} (i)))dw_u
		\\
		&+ \int_{s}^{t}\int_{|y|<c}F(\Phi_{s,u-}(\Phi_{0,s}(x,i),\Lambda_{0,s} (i)),y,\Lambda_{s,u-}(\Lambda_{0,s} (i)))\cpoim{u} 
		\\
		&+  \int_{s}^{t}\int_{|y|\geq c}G(\Phi_{s,u-}(\Phi_{0,s}(x,i),\Lambda_{0,s} (i)),y,\Lambda_{s,u-}(\Lambda_{0,s} (i)))\poim{u},
		\\
		\Lambda_{s,t} (\Lambda_{0,s} (i))=&\: i + \int_{s}^{t} \int_{\mathbb{R}} h\left(\Lambda_{s,u-}(\Lambda_{0,s} (i)), z\right) \Pi(d u, d z).
	\end{align*}
	Since the two component process $\left(\Phi_{s,t}(y,j),\Lambda_{s,t} (j)\right)$ is determined by the increments of type $w_{s+u}-w_{s}, N(s+u,A)-N(s,A), \Pi(s+u,B)-\Pi(s,B)$ for $u\geq0$, see, for example, Lemma 4.3.12 in \cite{applebaum2009levy}, so it is $\mathcal{G}_t$-measurable. Therefore the two component process $\left(\Phi_{s,t}(\Phi_{0,s}(x,i),\Lambda_{0,s} (i)),\Lambda_{s,t} (\Lambda_{0,s} (i))\right)$ is independent of $\mathcal{F}_s$. Meanwhile we have $\left(\Phi_{s,t}(\Phi_{0,s}(x,i),\Lambda_{0,s} (i)),\Lambda_{s,t} (\Lambda_{0,s} (i))\right)=\left(\Phi_{0,t}(x,i),\Lambda_{0,t} (i)\right)$ by the uniqueness of the solution. Hence, by Lemma \ref{mclem}, we have for $f$ is a bounded borel measurable function, 
	\begin{align*}
		\E \left(f(\left(\Phi_{0,t}(x,i),\Lambda_{0,t} (i)\right)|\mathcal{F}_s)\right) &= \E \left(f(\left(\Phi_{s,t}(\Phi_{0,s}(x,i),\Lambda_{0,s} (i)),\Lambda_{s,t} (\Lambda_{0,s} (i))\right))|\mathcal{F}_s)\right) 
		\\
		&=  \E \left(f(\left(\Phi_{s,t}(y,j),\Lambda_{s,t} (j)\right))|\mathcal{F}_s)\right) |_{y=\Phi_{0,s}(x,i),j=\Lambda_{0,s} (i)}.
	\end{align*}
The same argument shows that 
\begin{align*}
\E \left(f(\left(\Phi_{0,t}(x,i),\Lambda_{0,t} (i)\right)|\left(\Phi_{0,s}(x,i),\Lambda_{0,s} (i)\right)\right)
=  \E \left(f(\left(\Phi_{s,t}(y,j),\Lambda_{s,t} (j)\right))|\mathcal{F}_s)\right) |_{y=\Phi_{0,s}(x,i),j=\Lambda_{0,s} (i)}. 
\end{align*}
Therefore we obtain the Markov property of two-component process $\left(\Phi_{0,t}(x,i),\Lambda_{0,t} (i)\right)$.
\end{proof}

\bigskip

\subsection*{Acknowledgement}
This work was partially supported by WISE program (MEXT) at Kyushu University (YC), and by JST CREST Grant Number JPMJCR2115, Japan, and by JSPS KAKENHI Grant Number 22H01139 (HM).

\bigskip

\bibliographystyle{apalike} 
\bibliography{base_switch}

\begin{thebibliography}{}

\bibitem[Applebaum, 2009]{applebaum2009levy}
Applebaum, D. (2009).
\newblock {\em L{\'e}vy processes and stochastic calculus}.
\newblock Cambridge university press.

\bibitem[Balakrishnan et~al., 2017]{balakrishnan2017statistical}
Balakrishnan, S., Wainwright, M.~J., and Yu, B. (2017).
\newblock Statistical guarantees for the {EM} algorithm: From population to
  sample-based analysis.
\newblock {\em The Annals of Statistics}, 45(1):77--120.

\bibitem[Barndorff-Nielsen, 1997]{barndorff1997normal}
Barndorff-Nielsen, O.~E. (1997).
\newblock Normal inverse {G}aussian distributions and stochastic volatility
  modelling.
\newblock {\em Scandinavian Journal of statistics}, 24(1):1--13.

\bibitem[Chevallier and Goutte, 2017]{ChevallierGoutte2017}
Chevallier, J. and Goutte, S. (2017).
\newblock On the estimation of regime-switching {L}évy models.
\newblock {\em Studies in Nonlinear Dynamics \& Econometrics}, 21(1):3--29.

\bibitem[Cl\'{e}ment and Gloter, 2020a]{clement2020}
Cl\'{e}ment, E. and Gloter, A. (2020a).
\newblock Joint estimation for {SDE} driven by locally stable {L}\'{e}vy
  processes.
\newblock {\em Electron. J. Stat.}, 14(2):2922--2956.

\bibitem[Cl\'{e}ment and Gloter, 2020b]{CleGlo20}
Cl\'{e}ment, E. and Gloter, A. (2020b).
\newblock Joint estimation for {SDE} driven by locally stable {L}\'{e}vy
  processes.
\newblock {\em Electron. J. Stat.}, 14(2):2922--2956.

\bibitem[Dempster et~al., 1977]{dempster1977maximum}
Dempster, A.~P., Laird, N.~M., and Rubin, D.~B. (1977).
\newblock Maximum likelihood from incomplete data via the {EM} algorithm.
\newblock {\em Journal of the Royal Statistical Society: Series B
  (Methodological)}, 39(1):1--22.

\bibitem[Hahn et~al., 2010]{hahn2010markov}
Hahn, M., Fr{\"u}hwirth-Schnatter, S., and Sass, J. (2010).
\newblock {Markov chain Monte Carlo methods for parameter estimation in
  multidimensional continuous time Markov switching models}.
\newblock {\em Journal of Financial Econometrics}, 8(1):88--121.

\bibitem[Hibbah et~al., 2020]{hibbah2020mcmc}
Hibbah, E.~H., El~Maroufy, H., Fuchs, C., and Ziad, T. (2020).
\newblock {An MCMC computational approach for a continuous time state-dependent
  regime switching diffusion process}.
\newblock {\em Journal of Applied Statistics}, 47(8):1354--1374.

\bibitem[Jones et~al., 2014]{spuRs_book}
Jones, O., Maillardet, R., and Robinson, A. (2014).
\newblock {\em Introduction to scientific programming and simulation using
  {R}}.
\newblock The $R$ Series. CRC Press, Boca Raton, FL, second edition.

\bibitem[Kawai and Masuda, 2013]{kawai2013}
Kawai, R. and Masuda, H. (2013).
\newblock Local asymptotic normality for normal inverse {G}aussian {L}{\'e}vy
  processes with high-frequency sampling.
\newblock {\em ESAIM: Probability and Statistics}, 17:13--32.

\bibitem[Kim, 1994]{kim1994dynamic}
Kim, C.-J. (1994).
\newblock Dynamic linear models with {M}arkov-switching.
\newblock {\em Journal of econometrics}, 60(1-2):1--22.

\bibitem[Lange, 1995]{lange1995gradient}
Lange, K. (1995).
\newblock A gradient algorithm locally equivalent to the em algorithm.
\newblock {\em Journal of the Royal Statistical Society: Series B
  (Methodological)}, 57(2):425--437.

\bibitem[Mao and Yuan, 2006]{mao2006stochastic}
Mao, X. and Yuan, C. (2006).
\newblock {\em Stochastic differential equations with Markovian switching}.
\newblock Imperial college press.

\bibitem[Masuda, 2019]{masuda2019non}
Masuda, H. (2019).
\newblock Non-{G}aussian quasi-likelihood estimation of {SDE} driven by locally
  stable {L{\'e}vy} process.
\newblock {\em Stochastic Processes and their Applications}, 129(3):1013--1059.

\bibitem[Masuda, 2023]{Mas23}
Masuda, H. (2023).
\newblock Optimal stable {O}rnstein-{U}hlenbeck regression.
\newblock {\em Jpn. J. Stat. Data Sci.}, 6(1):573--605.

\bibitem[Skorokhod, 2009]{skorokhod2009asymptotic}
Skorokhod, A.~V. (2009).
\newblock {\em Asymptotic methods in the theory of stochastic differential
  equations}, volume~78.
\newblock American Mathematical Soc.

\bibitem[Xi and Zhu, 2017]{xi2017feller}
Xi, F. and Zhu, C. (2017).
\newblock On feller and strong feller properties and exponential ergodicity of
  regime-switching jump diffusion processes with countable regimes.
\newblock {\em SIAM Journal on Control and Optimization}, 55(3):1789--1818.

\bibitem[Yin and Xi, 2010]{yin2010stability}
Yin, G. and Xi, F. (2010).
\newblock Stability of regime-switching jump diffusions.
\newblock {\em SIAM Journal on Control and Optimization}, 48(7):4525--4549.

\end{thebibliography}

\end{document}